\documentclass[reqno]{amsart}
\usepackage{enumerate}
\usepackage{bbm}
\usepackage[margin=1.25in]{geometry}
\usepackage{amsmath,amsthm,amscd,amssymb}
\usepackage{latexsym}
\usepackage{color}
\usepackage{todonotes}
\usepackage{mathrsfs}
\usepackage{float}

\numberwithin{equation}{section}

\newtheorem{theorem}{Theorem}[section]

\newtheorem{corollary}[theorem]{Corollary}
\newtheorem{lemma}[theorem]{Lemma}
\newtheorem{proposition}[theorem]{Proposition}

\newtheorem{remark}[theorem]{Remark}

\newcommand{\N}{\mathbb{N}}
\newcommand{\R}{\mathbb{R}}

\newcommand{\calL}{\mathcal{L}}

\newcommand{\calO}{\mathcal{O}}

\newcommand{\st}{\mid}

\newcommand{\Exp}{\mathrm{E}}

\begin{document}

\title{Higher energy state approximations in the `Many Interacting Worlds' model}

\begin{abstract}
  In the `Many Interacting Worlds' (MIW) discrete Hamiltonian system approximation of 
  Schr\"odinger's wave equation, introduced in \cite{hall_2014}, convergence of ground states to the Normal ground state of the quantum harmonic oscillator, via Stein's method, in Wasserstein-$1$ distance with rate $\calO(\sqrt{\log N}/N)$ has been shown \cite{mckeague_levin_2016}, \cite{chen_thanh_2020}, \cite{mckeague_swan_2021}.  In this context, we construct approximate higher energy states of the MIW system, and show their convergence with the same rate in Wasserstein-$1$ distance to higher energy states of the quantum harmonic oscillator.  
  In terms of techniques, we apply the `differential equation' approach to Stein's method, which allows to handle behavior near zeros of the higher energy states.       
\end{abstract}



\author{Alex Loomis}
\address{Department of Mathematics\\
  University of Arizona\\
  621 N. Santa Rita Ave.\\
Tucson, AZ 85750, USA}
\email{{\tt atloomis@math.arizona.edu}}

\author{Sunder Sethuraman}
\address{Department of Mathematics\\
  University of Arizona\\
  621 N. Santa Rita Ave.\\
Tucson, AZ 85750, USA}
\email{{\tt sethuram@math.arizona.edu}}

\maketitle


\section{Introduction}
Consider the recent `Many Interacting Worlds' (MIW) discrete Hamiltonian system approximation \cite{hall_2014} of the Schr\"odinger equation
\begin{align}
  \label{Sch}
  i\hbar\partial_{t} \psi(t,x) = -\frac{\hbar^2}{2m} \partial_{xx}\psi(t,x) + x^2 \psi(t,x),
\end{align}
governing the wave function in a one dimensional quantum mechanical system.
Namely, following de Broglie-Bohm's interpretation \cite{bohm_1951} of \eqref{Sch}, write 
$\psi = \sqrt{P}e^{iS/\hbar}$ in terms of a probability density $P$ of the location of a particle and its momentum $\partial_{x} S$, where
\begin{align}
  \label{Sch_sys}
  \partial_t P + \partial_x \big(P(\partial_x S)\big) &= 0, \\
  \partial_t S + \frac{1}{2m} (\partial_x S)^2 + x^2 - \frac{\hbar^2}{4m}
  \Big({\partial_{xx} P}/{P} - \frac 12 \big({\partial_x P}/{P}\big)^2 \Big)
                                                      &= 0. \nonumber
\end{align}
As discussed in \cite{bohm_1951}, one may interpret \eqref{Sch_sys} via a Hamiltonian system where the average energy
\begin{align}
  \label{Bohm H}
  \bar H= \int P(x) \left( \frac{1}{2m} (\partial_x S)^2 + x^2
  + \frac{\hbar^2}{8m} \left(\frac{\partial_x P}{P}\right)^2 \right)\ dx.
\end{align}

In the MIW approach \cite{hall_2014}, a discretized Hamiltonian
\begin{align*}
  \bar H_{\text{MIW}}
  = \sum_{n=1}^N \frac{1}{2m} p_n^2
  + \sum_{n=1}^N x_n^2
  + \sum_{n=1}^N \frac{\hbar^2}{8m} \left(\frac{\partial_x P(x_n)}{P(x_n)}\right)^2
\end{align*}
is formulated where $(x_n)_{n=1}^N$, with boundary conditions $x_0=-\infty$ and $x_{N+1}=\infty$, are the locations of $N$ `world' particles, and
$(p_n)_{n=1}^N$ are their momenta.  
Therefore,
the average energy $\frac{1}{N}\bar H_{\text{MIW}}$ formally approximates $\bar H$.
We remark that
$U(x_n) = \big({\partial_x P(x_n)}/{P(x_n)}\big)^2$
represents a discretized form of Bohm's `quantum mechanical' potential
in reference to a `classical' $V(x_n) = x_n^2$ potential.

In \cite{hall_2014}, the sequence $(x_n)_{n=1}^N$ is taken
according to the mass $1/(N+1)$ quantiles of $P$,
that is when $\int_{x_n}^{x_{n+1}}P(u)\ du = 1/(N+1)$ for $0\leq n\leq N$.
For large $N$, assuming that $P$ is smooth,
the gaps $x_{n+1}-x_n$ should vanish, and so
  $P(x_n)
  \approx \frac{1}{(N+1)(x_{n+1} - x_{n})}$.

Moreover, the derivative
$\partial_x P(x_n)\approx \frac{P(x_{n+1}) - P(x_n)}{x_{n+1} - x_n}$.
Therefore,
\begin{align}
  \label{f'/f}
  \frac{\partial_x P(x_n)}{P(x_n)}
 &\approx (N+1)\partial_x P(x_n)(x_{n+1} - x_n) \nonumber\\
 &\approx (N+1)(P(x_{n+1})-P(x_n)) \approx \frac{1}{x_{n+1} - x_n} - \frac{1}{x_n - x_{n-1}}.
\end{align}
Note that interestingly the scale parameter $N$ cancels here.

Putting this approximation into $U(t,x_n)$,
and choosing units so that each `world' particle has unit mass,
and so that $\hbar^2 = 8m$,
we will work with the discrete Hamiltonian 
\begin{align}\label{eq:ur-hamiltonian}
  H_{\text{MIW}}
  = \sum_{n=1}^N \frac 12 p_n^2
  + \sum_{n=1}^N x_n^2
  + \sum_{n=1}^N \Big(\frac{1}{x_{n+1} - x_n} - \frac{1}{x_n - x_{n-1}}\Big)^2,
\end{align}
through which the
time evolution of the MIW system $(x_n=x_n(t))_{n=1}^N$, $(p_n = p_n(t))_{n=1}^N$, with $x_0=-\infty$, $x_{N+1}=\infty$ and $N\geq 2$,
is given by
\begin{align}
  \label{time_evo}
  p_n= \frac{d x_n}{dt} = \frac{\partial}{\partial p_n}{H_{\text{MIW}}}
  \quad\text{and}\quad
  \frac{d p_n}{dt} = -\frac{\partial}{\partial x_n}{H_{\text{MIW}}}.
\end{align}
We will say that a stationary configuration $(x_n)_{n=1}^N$ is one in which
$p_n=\frac{dx_n}{dt} = \frac{dp_n}{dt} = 0$
for all $1 \leq n \leq N$.

Figure \ref{fig:arbitrary} shows the
numerical simulation of the nonstationary dynamics of a five particle system
starting at the given initial positions,
and $x_0 = -\infty$
and $x_{6} = \infty$.
\begin{figure}
    \begin{center}
        \includegraphics[scale = 0.25]{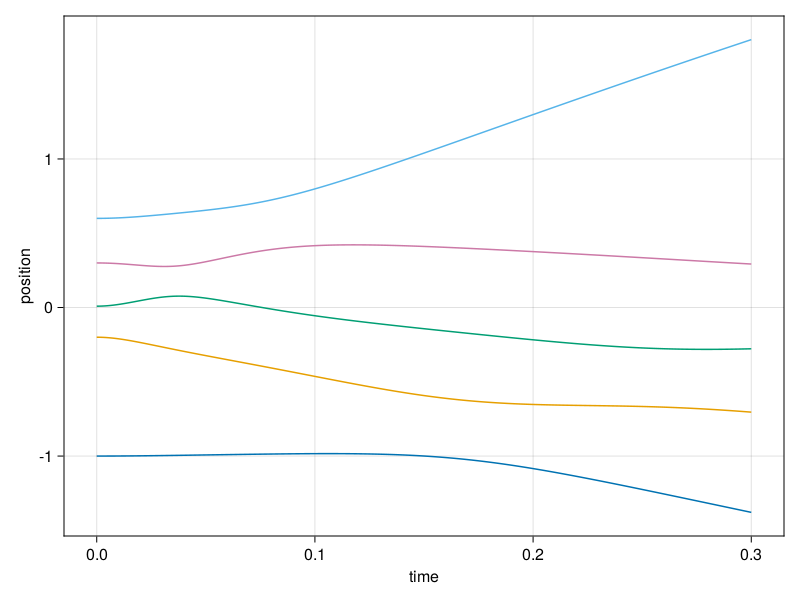}
        \caption{A simulation of MIW $N=5$ dynamics.}
        \label{fig:arbitrary}
      \end{center}
   \end{figure}

We comment that other discretizations of $U$,
involving more involved Taylor expansions,
have been considered in \cite{ghadimi_2018}.
We also mention higher dimensional MIW dynamics
have been considered in works \cite{herrmann_2018}, \cite{mckeague_swan_2021}, \cite{sturniolo_2018}.
However, in the following, we will focus on
the one dimensional formulation given above as it will allow some mathematical analysis.

\subsection{General aims}
Perhaps the main question at this point is whether indeed
the discrete evolution with respect to Hamiltonian \eqref{eq:ur-hamiltonian}
approximates the continuous one in terms of \eqref{Bohm H}.
In this generality, this is an open question.

However, previous work has considered certain `ground state' approximations.
A ground state configuration $(x_n)_{n=1}^N$ is an increasing sequence on which
$H_{\text{MIW}}$ is minimized.
In such a state,  the system is stationary:
$p_n=\partial_t x_n = 0$ for $1\leq n\leq N$.
In this time-independent setting, we observe
\begin{align}
  \label{time-indep-H}
  H_{\text{MIW}}=H =  \sum_{n=1}^N x_n^2
  + \sum_{n=1}^N \left(\frac{1}{x_{n+1} - x_n} - \frac{1}{x_n - x_{n-1}}\right)^2.
\end{align}
Following \cite{hall_2014}, consider
$A = \sum_{n=2}^N (x_{n-1}-x_n) \frac{1}{x_n - x_{n-1}} = -(N-1)$.
Summation by parts, noting $x_0=-\infty$, $x_{N+1} = \infty$ and $N\geq 2$, together with the Cauchy-Schwartz inequality, yields
\begin{align}\label{eq:ineq}
  \frac{A^2}{\sum_{n=1}^N x_n^2}
  \leq \sum_{n=1}^N {\Big(\frac{1}{x_{n+1} - x_n} - \frac{1}{x_n - x_{n-1}}\Big)}^2,
\end{align}
with equality if and only if
there is a constant $c$ such that, for all $1\leq n\leq N$,
\begin{align}\label{cond:eq}
  -x_n = c \Big(\frac{1}{x_{n+1} - x_n} - \frac{1}{x_n - x_{n-1}}\Big).
\end{align}

Hence,
$H
\geq  V+ A^2/V$,
where $V = \sum_{n=1}^N x_n^2$.  Optimizing over $V$, we get that $H\geq 2|A| = 2(N-1)$,
with equality exactly when
\eqref{cond:eq} holds
and  $V = N-1$.
Moreover, in this case,
by summing the squares of both sides of \eqref{cond:eq}, we have
$V
= c^2 \sum_{n=1}^N \big(\frac{1}{x_{n+1} - x_n} - \frac{1}{x_n - x_{n-1}}\big)^2
= c^2 \frac{A^2}{V}$,
and so $c^2 = 1$, as $V=|A|=N-1$.
Since the sign of $c$ determines only
whether the order of $(x_n)_{n=1}^N$ is reversed,
one may choose $c = 1$.

In particular, the corresponding sequence
$(x_n)_{n=1}^N$ satisfying
\begin{align*}
  \frac{1}{x_{n+1} - x_n} - \frac{1}{x_n - x_{n-1}} = -x_n,
\end{align*}with boundary conditions
$\frac{1}{x_2-x_1} = -x_1$ and $\frac{1}{x_N-x_{N-1}} = x_N$,
as $x_0=-\infty$ and $x_{N+1}=\infty$, is the unique increasing minimum of $H$, and is therefore a ground state.  Although the sequence values $x_n = x_n^{(N)}$ depend on $N$, to simplify notation, we will suppress the superscript in the following.

A question was posed in \cite{hall_2014} whether the empirical distribution of these ground states $(x_n)_{n=1}^N$ converged as $N\uparrow\infty$ to the standard Normal distribution, that is the ground state of the time-homogeneous Schr\"odinger's equation, namely that of the quantum harmonic oscillator.  This was recently resolved in the affirmative in \cite{mckeague_levin_2016, chen_thanh_2020}; see also \cite{mckeague_swan_2021}.

Indeed, we were inspired by the
convergence shown there in Wasserstein-$1$ distance, with upper and lower rates of order $\frac{\sqrt{log N}}{N}$, in
\cite{mckeague_levin_2016,chen_thanh_2020}, via a form of Stein's method (cf. surveys \cite{Chatterjee}, \cite{Ross}).

\subsection{Informal statement of results}
In this context, the goal of this paper is to study certain approximations of higher energy critical sequences
of $H$ in \eqref{time-indep-H}, 
and to investigate their convergence in Wasserstein-$1$ distance to higher energy states
of the continuous quantum harmonic oscillator system.
After the Normal ground state, the next continuous energy state of the quantum harmonic oscillator
is the Maxwellian distribution.
Higher continuous energy states, described later, can be defined
in terms of the Hermite polynomials.

We comment that in \cite{mckeague_2019},
empirical distributions of sequences, seen as ground states of related systems but
with different quantum potentials $U$,
were shown to converge in Wasserstein-$1$ distance to the Maxwellian; see also \cite{chen_wang_2023} which considers ground state sequences for systems with `Coulomb' potential $V$ and an associated quantum potential $U$.  However, as remarked in \cite{mckeague_2019}, to show validity of the MIW model,
one would like to show such convergences with respect to `critical' sequences
corresponding to critical points of $H$
to the Maxwellian and other higher energy continuum states.  
A difficulty is that it does not seem so facile to analyze the exact critical points of $H$, where $\nabla H = 0$,
other than that corresponding to the global minimum.  Our work in this respect will be to consider certain `approximations' of these critical points.

Informally, our results include the following:

\begin{itemize}
  \item[(a)] In Theorem \ref{thm:exist-unique},
    we show existence and uniqueness of specific potential approximations
    of critical points, which we will call MIW sequences.  
  \item[(b)] The `no-gap' property and spanning order of the MIW sequence,
    that is that $x_{n+1}-x_n$ vanishes
    and $x_1, x_N = \calO(\sqrt{\log N})$ as $N\uparrow\infty$,
    is given in Theorem \ref{thm:gaps}.
  \item[(c)] We show in Theorem \ref{thm:converge},
    via a form of Stein's method,
    that the empirical distributions of these sequences
    converge in Wasserstein-$1$ distance
    with rate $O(\sqrt{\log N}/N)$ to the Maxwellian
    and the other higher energy continuum state distributions.  
  \item[(d)]Finally, in Theorem \ref{thm:grad}, we show that the gradient of $H$ on these sequences vanishes in a sense pointwise as $N\uparrow\infty$, and therefore these sequences may be viewed as approximate `critical' points of $H$. 
\end{itemize}

\subsection{MIW sequences and higher energy functions}\label{sec:hdw}

The sequences we consider are of the following type.
Let $f$ be a smooth, nonnegative function on $\R$.
We will say an MIW sequence of $f$ is a
a strictly increasing sequence $(x_n)_{n=1}^N$
on $\R$, away from the zeroes of $f$, satisfying the relation
\begin{align}\label{eq:hdw-seq}
  \frac{1}{x_{n+1} - x_n} - \frac{1}{x_n - x_{n-1}}
  = \frac{f'(x_n)}{f(x_n)}
\end{align}
for $1 < n < N$.
The rationale for such a formulation is that \eqref{eq:hdw-seq}
is a type of `quantile' approximation, similar to \eqref{f'/f},
where $P'(x_n)/P(x_n)$ is replaced by $f'(x_n)/f(x_n)$.
Note that we have chosen $(x_n)_{n=1}^N$ as increasing,
although an alternate definition where the sequence is decreasing
could also be used as in \cite{mckeague_levin_2016}.

We will say that an MIW sequence $(x_n)_{n=1}^N$, contained in an possibly infinite interval $(a,b)$, satisfies the `left' and `right' boundary conditions at $x_0=a$ and $x_{N+1}=b$ if
$\frac{1}{x_2 - x_1} - \frac{1}{x_1 - a} = \frac{f'(x_1)}{f(x_1)}$ and
$\frac{1}{b - x_N} - \frac{1}{x_N - x_{N-1}} = \frac{f'(x_N)}{f(x_N)}$,
respectively hold and are well defined.
Note in the case $N=1$ that the sequence $(x_1)$ is well-defined
exactly when one or both of $a,b$ is finite.
With a right boundary condition of $b=\infty$,
by summing \eqref{eq:hdw-seq} over $n$ and taking the reciprocal, the MIW sequence satisfies
\begin{align}
\label{MIW-relation1}
x_{n+1} = x_{n} -\left(\sum_{k=n+1}^N \frac{f'(x_k)}{f(x_k)}\right)^{-1},
\end{align}
for $1\leq n<N$.  In this generality, it is not clear that such MIW sequences exist or are unique, even when $f$ is strictly positive!  See Appendix \ref{non-exist-sect} for some counter-examples.

However, in this article, we will be concerned with stable states $f$ of the quantum harmonic oscillator, for which we show MIW sequences may be constructed.
For $\ell\geq 0$, we define the functions
$f(x) = c p_\ell(x)^2 e^{-\frac 12 x^2}$
as higher energy functions, where $c \in \R^+$ is a normalization to make $f$ a probability density.  Here, $p_\ell$ is the $\ell$th Hermite polynomial,
$p_\ell(x) = (-1)^\ell e^{\frac 12 x^2} \frac{d^\ell}{dx^\ell} e^{-\frac 12 x^2}$.
Alternatively,
the Hermite polynomials may be specified in terms of the differential equation
$p_\ell''(x) = xp_\ell'(x) - \ell p_\ell(x)$.

We say the order of $f(x)$ is the order of $p_\ell(x)$, namely $\ell\geq 0$.
Note that $p_0 = 1$,
in which case $f(x) = \frac{1}{\sqrt{2\pi}} e^{-\frac 12 x^2}$
is the standard normal density.
As well, $p_1 = x^2$,
in which case $f(x) = \sqrt{\frac{8}{\pi}} x^2 e^{-\frac 12 x^2}$
is the Maxwellian distribution density.
We observe that the $\ell$th order higher energy function $f$ has exactly $\ell$ roots (Lemma \ref{lem:dbl-zro}), and is symmetric and log-concave (Lemma \ref{lem:higher-order}).

\subsection{Proof ideas for the main theorems}
\label{proof-sect}

Our method of constructing MIW sequences with respect to
$\ell$th order higher energy function $f$
in the proof of Theorem \ref{thm:exist-unique}
will be to create subsequences that lie entirely within
the $\ell+1$ regions of strict positivity of $f$.
These subsequences, with specified numbers of points in each region,
will be created so that their boundary conditions
allow them to be strung together to form the desired MIW sequence.

The strategy of this construction follows partly
the path taken in \cite{mckeague_levin_2016} with respect to the Normal density
(the order zero strictly positive higher energy function),
although there are new difficulties and differences,
since higher energy functions of order $\ell\geq 1$ have $\ell$ roots to be bridged.
We will define functions $(\chi_n(x))_{n\geq 1}$ inductively
on appropriate domains in the $\ell +1$ regions of strict positivity of $f$, where
$\chi_1(x) = x$, and
\[
  \frac{1}{\chi_{n+1}(x) - \chi_{n}(x)} - \frac{1}{\chi_n(x) - \chi_{n-1}(x)}
  = \frac{f'(\chi_n(x))}{f(\chi_n(x))},
\]
so that by adjusting the value $x$, the sequence $(\chi_n(x))_{n=1}^N$ will be an MIW sequence with appropriate left and right boundary conditions $\chi_0(x)$ and $\chi_{N+1}(x)$.
We derive uniqueness of the MIW sequence, with prescribed numbers of points in each region of strict positivity of $f$, via this construction, as sequences with different initial values, which determine the other points,
must satisfy different boundary conditions.

An ingredient for the proofs of the Wasserstein-$1$ convergence and stability in Theorems \ref{thm:converge} and \ref{thm:grad}, as well as intrinsic interest, will be to describe in Theorem \ref{thm:gaps} that the constructed MIW sequences have a 
`no gap' property in that the maximal gap between points $x_n$ and $x_{n+1}$ for $1\leq n\leq N-1$ vanishes as $N\uparrow\infty$.  This will be a consequence of the MIW recursion that $(x_n)_{n=1}^N$ satisfies, and that gaps near the roots and at extremities, which we show are the largest, vanish as the number of points in the strictly positive intervals diverge as $N\uparrow\infty$.   We will also specify $O\big(\sqrt{\log N}\big)$ asymptotics of $x_1, x_N$ as $N\uparrow\infty$, generalizing estimates in \cite{chen_thanh_2020} in the case of Normal density.

Proofs of convergence in Wasserstein-$1$ distance, with bounding rate of $\calO\big(\sqrt{\log N}/N\big)$, of the empirical measures of the MIW sequences constructed in Theorem \ref{thm:converge} to distributions with higher energy densities $f$
follow the general Stein's differential equation approach \cite{Chatterjee-Shao}, \cite{chen_2011}, \cite{stein_1986} as opposed to methods involving zero-bias distributions (cf. \cite{goldstein}) as in 
\cite{mckeague_levin_2016, chen_thanh_2020} with respect to the Normal density, and \cite{mckeague_2019} with respect to the Maxwellian density. See also \cite{mckeague_swan_2021}[Example 3.7] where a form of the differential Stein equation approach, fashioned to estimate directly between the empirical measure and the higher energy distribution, was used to show similar convergence with respect to the Normal density.  We mention  \cite{mckeague_swan_2021}[Theorem 3.4] contains abstract Wasserstein-$1$ distance bounds in this vein
between empirical measures of symmetric monotone sequences and distributions with density proportional to $b(x)e^{-\frac{1}{2}x^2}$ where $b$ is nonnegative with $b(x)>0$ for $x\neq 0$, although these do not apply in our general context.

A main difficulty to surmount is
that the higher energy densities of order $\ell\geq 1$ have roots where singularities arise in the method.
In this respect, we employ an `intermediate' continuous distribution as in \cite{mckeague_levin_2016, chen_thanh_2020} to bound the distance between it and the empirical measure.  We then apply a suitable form of the differential equation approach to bound the distance between the `intermediate' and higher energy distributions.  This requires a detailed analysis of Stein method bounding terms, near zeros of $f$ and at $\pm \infty$, that we provide.

Our tack, as in the existence and uniqueness arguments, is to separate $\R$ into intervals of strict positivity, bounded by the zeros of $f$,
and to bound Wasserstein-$1$ distances between empirical distributions and distributions with density proportional to $f$ restricted to each interval.  We obtain bounds in terms of of order $(x_N-x_1)/N$ where $x_1, x_N$ are the first and last points in the interval.
When put together, we obtain the desired rate of convergence in Wasserstein-$1$ distance of the whole sequence on $\R$.

Finally, by explicit computation of the gradient of $H$ on MIW sequences,
we will show in the proof of Theorem \ref{thm:grad} that,
for every point $t$ away from zeros of $f$,
the partial derivative of $H$ with respect to the closest point $x_n$
in the MIW sequence constructed
converges to $\big(t^2 - (f'(t)/f(t))^2 - 2[f'(t)/f(t)]'\big)'$ as $N \to \infty$.
However, this limit vanishes exactly when
$f$ solves $t^2 - (f'(t)/f(t))^2 - 2[f'(t)/f(t)]' = E$,
where $E$ is a constant,
that is the time-independent Schr\"odinger equation
or quantum harmonic oscillator system, $-4\psi''(t) + t^2\psi(t) = E\psi(t)$ with $\psi(t) = \sqrt{f(t)}$, in our units.
This happens exactly when $f$ is a higher energy density.
Therefore in this sense MIW sequences are approximately stable in the limit.

\subsection{Plan of the article}
After stating results and related remarks in Section \ref{results-sect},
we give the proofs of Theorems \ref{thm:exist-unique}, \ref{thm:gaps}, \ref{thm:converge}, and \ref{thm:grad}
in Sections \ref{exist-sect}, \ref{gap_sect}, \ref{converge-sect}, and \ref{grad-sect},
referring at times to an appendix with technical lemmas.

\section{Results}\label{results-sect}
Let $f$ be a higher energy function of order $\ell\geq 0$,
and when $\ell\geq 1$ let $r_1 < \cdots < r_\ell$ be its roots.
Choose $N_k \geq 1$ for $0 \leq k \leq \ell$ when $\ell\geq 1$, and $N=N_0\geq 2$ when $\ell=0$.
Denote the regions of strict positivity of $f$ by $R_0$ for the region $(-\infty, r_1)$,
by $R_k$ for the region $(r_k, r_{k+1})$ when $1 \leq k < \ell$,
and by $R_{\ell}$ for the region $(r_{\ell}, \infty)$.

In the following results, we will include the case $\ell=0$ with respect to the Normal densities to be complete although, as discussed earlier, parts of Theorems \ref{thm:exist-unique}, \ref{thm:gaps}, \ref{thm:converge} appear in \cite{mckeague_levin_2016}, \cite{chen_thanh_2020}, \cite{mckeague_swan_2021}.  We note the proof method for Theorem \ref{thm:converge} for $\ell=0$ differs from previous arguments.

We now state there is a unique MIW sequence matched to '$\mp \infty$' boundary conditions.
\begin{theorem}\label{thm:exist-unique}
  There is a unique MIW sequence $(x_n)_{n=1}^N$ of $f$ that satisfies
  the left boundary condition at $x_0=-\infty$,
  and the right boundary condition at $x_{N+1}=\infty$,
  with $N_k$ points that lie in the region $R_k$
  for each $0 \leq k \leq \ell$, and $N = \sum_{k=0}^{\ell}N_k$.
\end{theorem}

As a consequence, we deduce symmetry of MIW sequences when the number of points $N_k = N_{\ell-k}$ for $0\leq k\leq \ell$.  This condition automatically holds when $\ell=0$, that is when $f$ is the Normal density.

\begin{corollary}\label{cor:Normalsymmetry}
  When $N_k=N_{\ell-k}$ for $0\leq k\leq \ell$, the MIW sequence $(x_n)_{n=1}^N$ in Theorem \ref{thm:exist-unique} is symmetric in that $x_n = -x_{N-n+1}$ for $1\leq n\leq N$.
\end{corollary}

\begin{proof} Let $y_n = -x_{N-n+1}$ for $1\leq n\leq N$.  By symmetry of the higher energy density $f$ and that the number of points $N_k$ in $(r_{k}, r_{k+1})$ is the same as the number in reflected region $(r_{\ell-k}, r_{\ell-k+1})$, $(y_n)_{n=1}^N$ is also an MIW sequence, satisfying left and right boundary conditions at $\mp\infty$, with the same number of points $N$ as $(x_n)_{n=1}^N$.  Therefore, by the uniqueness part of Theorem \ref{thm:exist-unique}, $x_n = y_n$ for $1\leq n\leq N$.
\end{proof}

With respect to an MIW sequence, define now
\begin{align}
  \label{n(t) eq}
  n(t) = \max\big\{0\leq k\leq N: x_k \leq t\big\}
\end{align}
as the closest member of $(x_n)_{n=}^{N}$, with $x_0=-\infty$ and $x_{N+1}=\infty$, to the left of $t\in \R$,
or in other words $n(t)$ is the number of elements of $(x_n)_{n=1}^N$ less than or equal to $t$.  

We state that `gaps' $x_{n+1}-x_n$ vanish in the sequence, and give `spanning' orders of $x_1, x_N$.
\begin{theorem}\label{thm:gaps}
  Consider the MIW sequence $(x_n)_{n=1}^N$ in Theorem \ref{thm:exist-unique} such that the numbers $N_k\uparrow\infty$ for $0\leq k\leq \ell$.
  Then, for $t \in \R$,
  the difference $x_{n(t)+1} - x_{n(t)} \to 0$ as $N\uparrow\infty$.
  Consequentially, $\lim_{N \to \infty} x_{n(t)+1}^N = t$.  Moreover, at the extremities, there is a constant $c>0$ such that
  $c\sqrt{\log N}\leq |x_1|, x_N \leq c^{-1}\sqrt{\log N}$ for all large $N$.
\end{theorem}

Given two probability distributions $S$, $T$ on $\R$,
the Wasserstein-$1$ distance can be defined
$d(S, T) = \sup_h |\Exp_S[h] - \Exp_T[h]|$,
where $h$ ranges over the space of $1$-Lipschitz functions.  It can also be evaluated as $d(S,T)=\int_\R |F_S(x) - F_T(x)|dx$ where $F_S$, $F_T$ are the respective distribution functions.
By standard smoothing arguments, 
$d(S,T) = \sup_{h\in \mathcal{L}}|\Exp_S[h] - \Exp_T[h]|$,
where $\mathcal{L}$ is the space of differentiable $1$-Lipschitz functions $h$.
We note that Wasserstein-$1$ distance convergence implies weak convergence.

We state convergence in the Wasserstein-$1$ distance of the empirical distribution of the sequence to the distribution with higher energy density $f$.  Here, by convention, empty sums $\sum_{k=0}^{-1}$ vanish.

\begin{theorem}\label{thm:converge}
  Let $P$ denote the distribution on $\R$ with density $f$.
  Suppose the MIW sequence $(x_n)_{n=1}^N$ in Theorem \ref{thm:exist-unique} is such that $N_k = \lfloor N \int_{R_k} f(t)dt \rfloor$
  for $0\leq k< \ell$ and $N_\ell = N-\sum_{k=0}^{\ell-1} N_k \geq \lfloor N\int_{R_{\ell}}f(t)dt\rfloor$.
  Then, the empirical distribution $Q$ of $(x_n)_{n=1}^N$
  converges to $P$ with rate
  \begin{align*}
    d(Q,P) = \sup_{h\in \mathcal{L}}|\Exp_Q[h] - \Exp_{P_N}[h]|
    = \calO\left(\frac{\sqrt{\log N}}{N}\right).
  \end{align*}
  as $N\uparrow\infty$.
\end{theorem}

Finally, we give a sense in which the sequence is an approximate critical point of $H$, and in this sense an approximate higher energy state.

\begin{theorem}\label{thm:grad}
  The sequence $(x_n)_{n=1}^N$ in Theorem \ref{thm:exist-unique},
  such that the number of points $N_k\uparrow \infty$ for $0\leq k\leq \ell$,
  is stationary in the limit, in the sense that, evaluated on $(x_n)_{n=1}^N$,
  \begin{align*}
    \lim_{N \to \infty} \partial{x_{n(t)}} H = 0
  \end{align*}
  for each $t\in \R\setminus \{r_1, \ldots, r_\ell\}$, away from any zeros of $f$.
\end{theorem}

\subsection{Remarks}
We have the following comments.
\vskip .1cm

1. One might ask if one may construct MIW sequences for more general $f$,
not necessarily a higher energy function.
There may be application for instance to numerical density approximation.
Going through the proofs, one can see that the class of $f$'s can be made wider,
perhaps to `log-concave' $f$'s and to densities $f$ supported on intervals.
But, we leave this generalization to future work.  
\vskip .1cm

2. Also, we have taken $N_k\geq 1$ in Theorem \ref{thm:exist-unique},
which means every interval $(r_k, r_{k+1})$ contains points of the MIW sequence.
We believe with more work, and redefinition of the MIW sequence,
one can take some of the intervals empty.
But, given the application to approximation of the distributions
with higher energy density $f$'s, we have avoided this consideration.
\vskip .1cm

3. In \cite{chen_thanh_2020}, a lower bound of $\calO(\sqrt{\log N}/N)$
for the rate of Wasserstein-$1$ distance convergence of $Q$ to $P$
is shown in the Normal density, $\ell=0$, setting.
Convergence in Kolmogorov distance is also shown in \cite{chen_thanh_2020}.
It would be of interest to derive a lower bound
to pair with the upper bound given in Theorem \ref{thm:converge} when $\ell\geq 1$,
as well to consider other distances.
\vskip .1cm

4. Although the MIW sequences constructed are approximate higher energy states
for $H$ as stated in Theorem \ref{thm:grad},
they are not actual critical points of $H$, except when $\ell=0$ (see the expression for $\partial_{x_n}H$ in Lemma \ref{grad H calc}).
In fact, we observe in Lemma \ref{prop:graddiverge}, with respect to $\ell=1$,
that $\partial_{x_n} H$ diverges at $x_n=x_{n(0)+1}$
as $N\uparrow\infty$ with rate $x_{n(0)+1}^{-3}$.
In Lemma \ref{lem:Max0}, we show that $x_{n(0)+1}$ is of order $N^{-1/3}$,
from which we see $\partial_{x_{n(0)}+1}H$ diverges at order $N$.
This is not a contradiction of Theorem \ref{thm:gaps},
as $t=0$ is a zero of the Maxwellian density.

In Figure \ref{fig:maxwellian}, the $N=5$ initial positions are given by the Maxwellian $(\ell=1)$ MIW sequence
with three negative and two positive points.  One can compare the time evolution of this sequence with that in Figure \ref{fig:arbitrary} which follows the dynamics starting from a different configuration with the same energy $H$.  There is some indication that in a small time-scale that the MIW sequence is somewhat stationary (flatter trajectories).
It would be of interest to understand better
how well the evolved `worlds' $(x_n(t))_{n=1}^N$
approximate the initial MIW sequence positions $(x_n=x_n(0))_{n=1}^N$.
\vskip .1cm

5. Finally, in \cite{hall_2014},
MIW discrete Hamiltonian approximations of Schr\"odinger's equation
with non-Harmonic potentials $V$ are also discussed.  
From the proof discussion in Sections \ref{proof-sect} and \ref{grad-sect}
we can infer more generally that convergence of the empirical distribution of MIW sequences
to critical points should happen when
$V(t) - (f'(t)/f(t))^2 -2[f'(t)/f(t)]' = E$,
the associated time-independent Schr\"odinger's equation.
It would be of interest to investigate more carefully these convergences in this context.

  \begin{figure}
   \begin{center}
        \includegraphics[scale = 0.25]{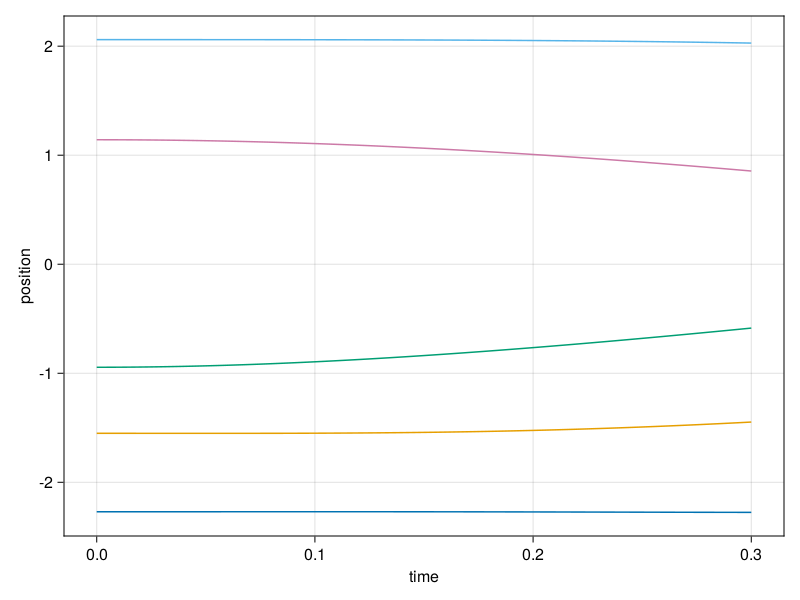}
        \caption{MIW $N=5$ dynamics starting from the Maxwellian $\ell=1$ approximating state $(x_n)_{n=1}^5$ with $3$ negative and $2$ positive values.}
        \label{fig:maxwellian}
              \end{center}
\end{figure}

\section{Proof of Theorem \ref{thm:exist-unique}:  Existence and uniqueness of MIW sequences}
\label{exist-sect}

We will begin by defining
an increasing sequence of functions $(\chi_n(x))_{n\geq 1}$, depending on an value $x\in \R$,
whose evaluations satisfy the left boundary condition at $-\infty$,
and the MIW recursion relation \eqref{eq:hdw-seq}.
Here, the parameter $x$ will correspond to the first element of the sequence, $\chi_1(x)=x$.
We will not consider the right boundary condition in defining this sequence,
but will later choose $x$ so that the sequence
satisfies the right boundary condition at $\infty$.

These sequences, depending on the position of $x$, will lie entirely in a single region
in which $f$ is strictly positive.
In order to create a sequence with points
in {\it every} strictly positive region of $f$,
we will use the single region existence to create a sequence in each region.
The left boundary conditions will be chosen
to allow the concatenation of these sequences
to satisfy the MIW recursion relation.

We will then work from right to left.
We choose the parameter of the last (say $\ell +1$th) sequence
so that the $\ell +1$th sequence satisfies the $\infty$-right boundary condition.
Then, we choose the parameter of the $\ell$th sequence
so that the concatenation of the $\ell$th and $\ell+1$th sequences
obeys the MIW recursion relation.
We continue by choosing the parameter of the $(\ell-1)$th sequence
so that the concatenation of the $(\ell-1)$th and $\ell$th sequences
obeys the MIW recursion relation, and so on to the initial sequence, which a priori satisfies the $-\infty$-left boundary condition.

Uniqueness will be a consequence of the property that
$\frac{d\chi_{n+1}}{d\chi_n} \geq 1$.
This relation ensures that two sequences
in the same strictly positive region of $f$
with different initial values
satisfy different right boundary conditions.

\subsection{Existence of the functions $\chi_n$}
\label{sect3.1}

In the first step, we wish to define a family of functions $\chi_n(x)$
that correspond to an MIW sequence starting at $x\in (r, b_0)$,
with a left boundary point of $a(x)$, leaving the right boundary undetermined.  Here, $r$ is a root of $f$ or $r = -\infty$, and $b_0$ is the next root of $f$ or $b_0=\infty$ if there are no subsequent roots.  It will be convenient to let $a(x): (r,b_0)\rightarrow [-\infty, \infty)$ be a differentiable function such that $a(x)<x$ and $0\leq \frac{da(x)}{dx} \leq 1$. We will also allow the choice of $a(x) \equiv -\infty$,
with the convention in this case that $\frac{1}{x - a(x)} = 0$ for $x\in (r, b_0)$.

In the second step, we will match the final elements of these MIW sequences in $(r, b_0)$ as left boundary points of MIW sequences in subsequent intervals.  In this way, $a(x)$ for each sequence will be later identified.

For $x \in (r, b_0)$, define
$\chi_1(x) = x$.
We now construct functions $\chi_n(x)$ for $n\geq 1$ which will form the next elements
of an MIW sequence.
We want to choose $b_1$ so that $(r, b_1)$
is the maximal domain of definition of $\chi_2(x)$.

\begin{lemma}
  \label{lem:3.1}
  The minimum
  \[
    b_1 = \min \Big\{x \in (r, b_0)
    \st \frac{1}{x - a(x)} + \frac{f'(x)}{f(x)} = 0\Big\}
  \]
  is well-defined.  Also, $\frac{1}{x-a(x)} + \frac{f'(x)}{f(x)} >0$ for $x\in (r, b_1)$, 
  $\lim_{x\downarrow r} \frac{1}{x-a(x)} + \frac{f'(x)}{f(x)} = \infty$, and
  $\lim_{x\uparrow b_1} \frac{1}{x-a(x)} + \frac{f'(x)}{f(x)} = 0$.
\end{lemma}

\begin{proof}
  Because $0 \leq \frac{da(x)}{dx} \leq 1$ and $a(x)<x$,
  the quotient $\frac{1}{x - a(x)}$ is a decreasing, positive function of $x$.
  By Lemma \ref{lem:3.5}, the limit $\lim_{x \downarrow r} \frac{f'(x)}{f(x)} = \infty$.
  So, the limit of the sum
  $
  \lim_{x \downarrow r} \frac{1}{x - a(x)} + \frac{f'(x)}{f(x)} = \infty.
  $
  
  Because $b_0$ is either a root of $f$ or is infinite, again by Lemma \ref{lem:3.5}.
  the limit $\lim_{x \uparrow b_0} \frac{f'(x)}{f(x)} = -\infty$.
  Thus, the limit of the sum
  $
    \lim_{x \uparrow b_0} \frac{1}{x - a(x)} + \frac{f'(x)}{f(x)} = -\infty.
  $
 
  Then, there is a point $b_1 \in (r, b_0)$ such that
  \[
    \lim_{x\uparrow b_1} \frac{1}{x-a(x)} + \frac{f'(x)}{f(x)} = \frac{1}{b_1 - a(b_1)} + \frac{f'(b_1)}{f(b_1)} = 0.
  \]
  Because $\frac{f'(x)}{f(x)}$ is strictly decreasing by Lemma \ref{lem:higher-order}, we have that
  $\frac{1}{x - a(x)} + \frac{f'(x)}/{f(x)}$
  is strictly decreasing for $x \in (r, b_0)$, the value $b_1$ is the unique such point.

  Moreover, by this construction, $\frac{1}{x-a(x)} + \frac{f'(x)}{f(x)}>0$ on $(r, b_1)$.
\end{proof}

We wish to mirror the recursion relation
$x_2 = x_1 + \left(\frac{1}{x_1 - a(x)} + \frac{f'(x_1)}{f(x_1)}\right)^{-1}
$ of an MIW sequence.
Thus, we define $\chi_2(x)$ on $(r, b_1)$ as
\[
  \chi_2(x) = x + \left(\frac{1}{x - a(x)} + \frac{f'(x)}{f(x)}\right)^{-1}.
\]

\noindent{\bf Base.}
In order to make an inductive definition for higher values of $n$,
we will now note several properties of $\chi_1(x)$ and $\chi_2(x)$.

\begin{lemma}
  Define $b_0$, $b_1$, $\chi_1$, and $\chi_2$ as above.
  The following hold:
  \begin{enumerate}
    \item $r < b_1 < b_0$,
    \item $\chi_1(x)$ is continuous for $x \in (r, b_0)$,
    \item $\chi_2(x)$ is continuous on $x \in (r, b_1)$,
    \item $\chi_1(x) < \chi_2(x)$ for all $x \in (r, b_1)$,
    \item $\lim_{x \downarrow r} \chi_2(x) = r$, and
    \item $\lim_{x \uparrow b_1} \chi_2(x) = \infty$.
  \end{enumerate}
\end{lemma}

\begin{proof}
  The first bullet holds because $b_1 \in (r, b_0)$.
  The second because $\chi_1(x) = x$.
  The third holds by inspection as $b_1$ was chosen
  to make $(r, b_1)$ the maximal domain
  on which $\chi_2$ is well-defined.
  The fourth holds because, after rearrangement, and noting Lemma \ref{lem:3.1},
  \[
    \frac{1}{\chi_2(x) - \chi_1(x)}
    = \frac{f'(x)}{f(x)} + \frac{1}{\chi_1(x) - a(x)}
    > 0.
  \]
  The fifth holds because
  $\lim_{x \downarrow r} \chi_1(x) = r$, and
  \[
    \lim_{x \downarrow r}
    \frac{1}{\chi_2(x) - \chi_1(x)}
    = \lim_{x \downarrow r} \frac{f'(x)}{f(x)} + \frac{1}{\chi_1(x) - a(x)}
    = \infty
  \]
  by Lemma \ref{lem:3.1}.
  The sixth holds because
  $\lim_{x \uparrow b_1} \chi_1(x) = b_1$ is finite, and
  \[
    \lim_{x \uparrow b_1}
    \frac{1}{\chi_2(x) - \chi_1(x)}
    = \lim_{x \uparrow b_1} \frac{f'(x)}{f(x)} + \frac{1}{\chi_1(x) - a(x)}
    = 0,
  \]
  by Lemma \ref{lem:3.1} again.
\end{proof}

\noindent{\bf Induction.}
Suppose we have successfully defined $\chi_{n-1}(x)$ and $\chi_n(x)$,
on the domain $(r, b_{n-1})$,
as we have done for $n = 2$.
Suppose that for some $b_{n-2}$, $b_{n-1} \in \R$,
\begin{enumerate}
  \item $r < b_{n-1} < b_{n-2}$,
  \item $\chi_{n-1}(x)$ is continuous on $x \in (r, b_{n-2})$,
  \item $\chi_n(x)$ is continuous on $x \in (r, b_{n-1})$,
  \item $\chi_{n-1}(x) < \chi_n(x)$ for all $x \in (r, b_{n-1})$,
  \item $\lim_{x \downarrow r} \chi_n(x) = r$, and
  \item  $\lim_{x \uparrow b_{n-1}} \chi_n(x) = \infty$.
\end{enumerate}
We have already shown that these hold for the base case of $n = 2$.

We wish now to define $b_n$ and the other items to complete induction.
To begin with, we shall evaluate the limits of
$\frac{1}{\chi_n(x) - \chi_{n-1}(x)} + \frac{f'(\chi_n(x))}{f(\chi_n(x))}$
at the boundary points of its domain to allow for an intermediate value argument.

\begin{lemma}\label{lem:left-lim}
  The limit
  \[
    \lim_{x \downarrow r}
    \frac{1}{\chi_n(x) - \chi_{n-1}(x)} + \frac{f'(\chi_n(x))}{f(\chi_n(x))}
    = \infty.
  \]
\end{lemma}

\begin{proof}
  By assumption, $\chi_{n-1}(x) < \chi_n(x)$ in a neighborhood of $r$, and $\lim_{x \downarrow r} \chi_n(x) =\lim_{x\downarrow r} \chi_{n-1}(x)= r$.  Then,
  $\lim_{x \downarrow r} \frac{1}{\chi_n(x) - \chi_{n-1}(x)} = \infty$.
  Note also $\lim_{x \downarrow r} \frac{f'(\chi_n(x))}{f(\chi_n(x))} = \lim_{x\downarrow r} \frac{f'(x)}{f(x)} = \infty$, by Lemma \ref{lem:3.5}.  The desired limit holds by adding the limits.
\end{proof}

Note that the range of $\chi_n(x)$ is $(r,\infty)$ for $x\in (r, b_{n-1})$.  For $r<s<\infty$, let 
$$\chi_n^{-1}(s) = \inf\{r<z<b_{n-1}: \chi_n(z) = s\}$$ 
be the generalized inverse. If $s=\infty$, let $\chi_n^{-1}(s)=b_{n-1}$.    Note also, as $b_0$ is the next root of $f$ after $r$, or $b_0=\infty$ if no subsequent root exists, that $\chi_n^{-1}(b_0)\leq b_{n-1}$, and $f'(\chi_n(x))/f(\chi_n(x))$ is continuous for $x\in (r, \chi_n^{-1}(b_0))$.   

\begin{lemma}\label{lem:right-lim}
  The limit
  \[
    \lim_{x \uparrow \chi_n^{-1}(b_0)}
    \frac{1}{\chi_n(x) - \chi_{n-1}(x)} + \frac{f'(\chi_n(x))}{f(\chi_n(x))}
    = -\infty.
  \]
\end{lemma}

\begin{proof}
  The limit $\lim_{x\uparrow \chi_n^{-1}(b_0)} \frac{1}{\chi_n(x) - \chi_{n-1}(x)} = \frac{1}{b_0 - \chi_{n-1}(\chi_n^{-1}(b_0))}<\infty$ as $\chi_n(x)>\chi_{n-1}(x)$ are continuous for $x\in (r, b_{n-1})$ by assumption.  However, by Lemma \ref{lem:3.5}, we have that
  $\lim_{x\uparrow \chi_n^{-1}(b_0)}\frac{f'(\chi_n(x))}{f(\chi_n(x))} = -\infty$, yielding the result.
\end{proof}

We now define $b_n$. 

\begin{lemma}\label{lem:this-bn}
  The value $b_n$ is well-defined and given by
  \[
    r<b_n = \min \left\{ x \in (r, \chi_n^{-1}(b_0)) \st
      \frac{1}{\chi_n(x) - \chi_{n-1}(x)} + \frac{f'(\chi_n(x))}{f(\chi_n(x))}
    = 0 \right\}< b_{n-1}.
  \]
  Also, $\frac{1}{\chi_n(x) - \chi_{n-1}(x)} + \frac{f'(\chi_n(x))}{f(\chi_n(x))}>0$ for $x\in (r, b_n)$, and $\lim_{x\uparrow b_n} \frac{1}{\chi_n(x) - \chi_{n-1}(x)} + \frac{f'(\chi_n(x)}{f(\chi_n(x))} = 0$.
\end{lemma}

\begin{proof}
  By Lemma~\ref{lem:left-lim}, the limit
  $\lim_{x \downarrow r}
  \frac{1}{\chi_n(x) - \chi_{n-1}(x)} + \frac{f'(\chi_n(x))}{f(\chi_n(x))}
  = \infty$.
  As well, by Lemma~\ref{lem:right-lim}, the limit
  $\lim_{x \uparrow \chi_n^{-1}(b_0)}
  \frac{1}{\chi_n(x) - \chi_{n-1}(x)} + \frac{f'(\chi_n(x))}{f(\chi_n(x))}
  = -\infty$.
  Thus, by continuity of $\frac{1}{\chi_n(x) - \chi_{n-1}(x)} + f'(\chi_n(x))/f(\chi_n(x))$ for $x\in (r, \chi_n^{-1}(b_0))$, following by assumption, there is a least intermediate point $r < b_n < \chi_n^{-1}(b_0)\leq b_{n-1}$, where
  $$\lim_{x\uparrow b_n} \frac{1}{\chi_n(x) - \chi_{n-1}(x)}
  + \frac{f'(\chi_n(x))}{f(\chi_n(x))} = \frac{1}{\chi_n(b_n) - \chi_{n-1}(b_n)} + \frac{f'(\chi_n(b_n)}{f(\chi_n(b_n))} = 0.$$
  Also, by construction, $\frac{1}{\chi_n(x) - \chi_{n-1}(x)} + \frac{f'(\chi_n(x))}{f(\chi_n(x))}>0$ for $x\in (r, b_n)$.  
\end{proof}

We are now ready to define $\chi_{n+1}(x)$ for $x\in (r, b_n)$.
Given
$\frac{1}{\chi_n(x) - \chi_{n-1}(x)} + \frac{f'(\chi_n(x))}{f(\chi_n(x))} > 0$
on $(r, b_n)$ by Lemma \ref{lem:this-bn}, define
\begin{align}
  \label{chi-n+1}
  \chi_{n+1}(x) = \chi_n(x)
  + \left( \frac{1}{\chi_n(x) - \chi_{n-1}(x)}
  + \frac{f'(\chi_n(x))}{f(\chi_n(x))} \right)^{-1}.
\end{align}

\begin{lemma}\label{lem:cont-xn}
  The function $\chi_{n+1}(x)$ is continuous and $\chi_{n+1}(x)>\chi_n(x)$ on $(r, b_n)$.
\end{lemma}

\begin{proof}
  On the domain $(r, b_n)\subset (r, b_{n-1})$,
  the functions $\chi_{n-1}(x)$ and $\chi_n(x)$ are both continuous and $\chi_{n-1}(x) < \chi_n(x)$ by assumption.
  Hence, $\frac{1}{\chi_n(x) - \chi_{n-1}(x)}$
  is continuous on $(r, b_n)$.
  By the definition of $b_n$, $\frac{f'(\chi_n(x))}{f(\chi_n(x))}$ is continuous on $(r, b_n)$.  Moreover, by Lemma \ref{lem:this-bn}, $\frac{1}{\chi_n(x) - \chi_{n-1}(x)} + \frac{f'(\chi_n(x))}{f(\chi_n(x))}>0$
  on $(r, b_n)$.
  Thus, its reciprocal is continuous and positive on $(r, b_n)$.  Therefore, $\chi_{n+1}(x)$ is continuous and $\chi_{n+1}(x)>\chi_n(x)$ on $(r, b_n)$.
\end{proof}

We now wrap up the definition of $\chi_{n+1}$ by showing
it satisfies the remaining properties required to complete the induction.

\begin{lemma}\label{lem:next-left}
  The limits
  $\lim_{x \downarrow r} \chi_{n+1}(x) = r$ and $\lim_{x\uparrow b_n}\chi_{n+1}(x)=\infty$.
\end{lemma}

\begin{proof}
  By assumption, $\lim_{x \downarrow r} \chi_n(x)
  = r$.  
  Then, by Lemma~\ref{lem:left-lim},
  $\lim_{x \downarrow r} \frac{f'(\chi_n(x))}{f(\chi_n(x))}
  = \infty$.  Since $\chi_n(x)>\chi_{n-1}(x)$ for $x\in (r, b_n)\subset (r, b_{n-1})$ by assumption,
  we have that
  $\lim_{x\downarrow r} \big(\frac{1}{\chi_n(x) - \chi_{n-1}(x)}
  + \frac{f'(\chi_n(x))}{f(\chi_n(x))} \big)^{-1} = 0$.  Hence, noting the definition of $\chi_{n+1}(x)$, the first limit holds.

  Next, as $\chi_n(x)$ is continuous at $x=b_n\in (r, b_{n-1})$ by assumption,
  $\lim_{x\uparrow b_n}\chi_n(x) = \chi_n(b_n)<\infty$.  Also, by Lemma \ref{lem:this-bn},
  $\lim_{x\uparrow b_n}\big(\frac{1}{\chi_n(x) - \chi_{n-1}(x)}
  + \frac{f'(\chi_n(x))}{f(\chi_n(x))} \big)^{-1} = \infty$.  The second limit holds as a consequence.
\end{proof}

For convenience, we summarize what we have shown via Lemmas \ref{lem:this-bn}, \ref{lem:cont-xn} and \ref{lem:next-left} in one place.

\begin{proposition}
  \label{prop:summary}
  The following properties hold:
  \begin{enumerate}
    \item $r < b_n < b_{n-1}$,
    \item $\chi_n(x)$ is continuous on $x \in (r, b_{n-1})$,
    \item $\chi_{n+1}(x)$ is continuous on $x \in (r, b_n)$,
    \item $\chi_n(x) < \chi_{n+1}(x)$ for all $x \in (r, b_n)$,
    \item $\lim_{x \to r} \chi_{n+1}(x) = r$, and
    \item $\lim_{x \to b_n} \chi_{n+1}(x) = \infty$,
  \end{enumerate}
  where recall that $r \in \R$ is a root of $f$, or $r = -\infty$, and $b_0$ is the next root after $r$ of $f$, or if there no subsequent root, $b_0=\infty$.
\end{proposition}

\subsection{Properties of the functions $\chi_n$}

In the following, let $N\geq 1$.  Recall the definition of MIW sequence \eqref{eq:hdw-seq}.

\begin{lemma}\label
  {lem:left-bc}
  For $x \in (r, b_N)$,
  the sequence $(\chi_n(x))_{n=1}^{N}$
  is an MIW sequence,  satisfying the left boundary condition
  with boundary point $a(x)$, and the right boundary condition
  with boundary point $\chi_{N+1}(x)$.
  Moreover, for $x = b_N$,
  the sequence $(\chi_n(x))_{n=1}^N$
  satisfies the right boundary condition
  with boundary $\infty$.

\end{lemma}

\begin{proof}
  By Lemma~\ref{lem:cont-xn}, the sequence is increasing, and
  the defining relation of $\chi_{n+1}(x)$ in \eqref{chi-n+1} is the MIW recursion relation.  Therefore, $(\chi_n(x))_{n=1}^N$ is an MIW sequence.

  Rearranging the form of $\chi_2(x)$ in \eqref{chi-n+1}, we obtain
  $\frac{1}{\chi_2(x) - \chi_1(x)} - \frac{1}{\chi_1(x) - a(x)}
  = \frac{f'(\chi_1(x))}{f(\chi_1(x))}$,
  which is the left boundary condition at $a(x)$.

  Also, by definition of $\chi_{N+1}(x)$ in \eqref{chi-n+1}, we have
  $\frac{1}{\chi_{N+1}(x) - \chi_N(x)} - \frac{1}{\chi_N(x) - \chi_{N-1}(x)}
  = \frac{f'(\chi_N(x))}{f(\chi_N(x))}$,
  which is the right boundary condition at $\chi_{N+1}(x)$.

  Finally, the point $b_N$ is defined (Lemma \ref{lem:this-bn}) so that by continuity the relation
  $$-\frac{1}{\chi_N(b_N) - \chi_{N-1}(b_N)}
  = \frac{f'(\chi_N(b_N))}{f(\chi_N(b_N))}$$
  holds,
  which is the right boundary condition at $\infty$.
\end{proof}

The above development allows to construct an MIW sequence, with respect to higher energy functions, albeit in a single strictly positive region of $f$, on $(r, b_0)$.

\begin{lemma}\label{lem:chi-bdd}
  The sequence $(\chi_k(b_N))_{k=1}^N$ is an MIW sequence of $f$, bounded in the interval $(r, b_0)$, satisfying the left boundary condition $a(x)$ and right boundary condition $\infty$.
  Moreover, 
  For $n\geq 1$ and $x \in (r, b_n]$,
  \[
    r < \chi_1(x) < \cdots < \chi_n(x) < b_0.
  \]
  We note $\chi_{n+1}(x)$ is defined for $x\in (r, b_n)$,
  however it is not bounded above,
  as $\lim_{x \uparrow b_n} \chi_{n+1}(x) = \infty$.
\end{lemma}

\begin{proof}
  By Lemma~\ref{lem:left-bc}, $(\chi_k(b_N))_{k=1}^N$ is an MIW sequence,
  satisfying the left and right boundary conditions.

  Recall that $r<x=\chi_1(x)<\chi_k(x)$ for $2\leq k\leq n$ on the interval $(r, b_n]$ (cf. Proposition \ref{prop:summary}).   The lower bound in the second statement follows as a consequence.

  For the upper bound, suppose $\chi_n(x_0) = b_0$ for some $x_0 \in (r, b_n]$.  Note that
  $\chi_n(x_0) - \chi_{n-1}(x_0)>0$, as $r<x_0\leq b_n<b_{n-1}$, by Proposition \ref{prop:summary} again.  Also, $\lim_{x\uparrow b_0} f'(x)/f(x) = \infty$ by Lemma \ref{lem:3.5}.
  Then, 
  \[
    \chi_{n+1}(x_0) = \lim_{x\uparrow x_0} \chi_{n+1}(x)
    = \lim_{x\uparrow x_0} \chi_n(x)
    + \left( \frac{1}{\chi_n(x) - \chi_{n-1}(x)}
    + \frac{f'(\chi_n(x))}{f(\chi_n(x))} \right)^{-1}
    = \chi_n(x_0),
  \]
  which contradicts the monotonicity $\chi_{n+1}(x)>\chi_n(x)$ for $r<x\leq b_n$ (cf. Proposition \ref{prop:summary}). 
  Thus, $\chi_n(x) \neq b_0$.
  However, by Proposition \ref{prop:summary}, the limit $\lim_{x \downarrow r} \chi_n(x) = r < b_0$ and $\chi_n(x)$ is continuous for $x\in (r, b_n]\subset (r, b_{n-1})$.
  Therefore, $\chi_k(x)<\chi_n(x) < b_0$ for $1\leq k\leq n-1$ and $x\in (r, b_n]$.
\end{proof}

When $f$ is the Normal density, we have indeed constructed, as a direct consequence of Lemma \ref{lem:chi-bdd}, the desired MIW sequence, as also done in \cite{mckeague_levin_2016}.

\begin{corollary}[Normal $\ell =0$]
  \label{cor:Normal-exist}
  Let $f(x) = \frac{1}{\sqrt{2\pi}} e^{-\frac 12 x^2}$.
  Choose $a(x) \equiv -\infty$.  As $f$ has no roots, $b_0=\infty$.
  The sequence $(\chi_n(b_N))_{n=1}^N$, with $N\geq 2$,
  is an MIW sequence of $f$
  satisfying both the left boundary condition at $-\infty$,
  and right boundary condition at $\infty$.
\end{corollary}

In Section \ref{sec:exist}, we will show how to construct sequences for higher energy $f$ when $\ell\geq 1$.  We now give two results that will be used
to help define the choice of $a(x)$ in this respect.

\begin{lemma}\label{lem:lcderivative}
  For $1\leq m \leq n$ and $x \in (r, b_{n-1})$, the derivative
  \begin{align*}
    \frac{d\chi_n}{d\chi_m}(x) \geq 1.
  \end{align*}
  In particular, as $\chi_1(x)=x$, we have $\frac{d\chi_n}{dx}\geq 1$ for $x\in (r, b_{n-1})$.
\end{lemma}

\begin{proof}
  We will show first that $d\chi_{m+1}(x)/d\chi_m(x)\geq 1$.
  As a base case, when $m = 1$,
  the derivative
  \begin{align*}
    \frac{d\chi_2}{d\chi_1}
    &= \frac{d}{d\chi_1} \left[
      \chi_1 + \left( \frac{1}{\chi_1 - a(x)}
      + \frac{f'(\chi_1)}{f(\chi_1)} \right)^{-1}
    \right] \\
    &= 1 - \left( \frac{1}{\chi_1 - a(x)}
    + \frac{f'(\chi_1)}{f(\chi_1)} \right)^{-2}
    \frac{d}{d\chi_1} \left( \frac{1}{\chi_1 - a(x)}
    + \frac{f'(\chi_1)}{f(\chi_1)} \right).
  \end{align*}
  The squared term is positive.
  It will be enough to show that
  $\frac{d}{d\chi_1} \left( \frac{1}{\chi_1 - a(x)}
  + \frac{f'(\chi_1)}{f(\chi_1)} \right)$ is negative.
  Because $\frac{da(x)}{dx} \leq 1$, and $\chi_1(x) = x$,
  \begin{align*}
    \frac{d}{d\chi_1} \frac{1}{\chi_1 - a(x)}
    = -\left(x - a(x)\right)^{-2} \left(1 - \frac{da(x)}{dx}\right)
    \leq 0.
  \end{align*}
  On the other hand, by log-concavity of $f$ (Lemma \ref{lem:higher-order}), the derivative
  $\frac{d}{d\chi_1} \frac{f'(\chi_1)}{f(\chi_1)}
  = \frac{d}{dx} \frac{f'(x)}{f(x)}$ is negative, finishing the base case.

  For $m\geq 2$, the derivative
  \begin{align*}
    \frac{d\chi_{m+1}}{d\chi_m} 
    &= \frac{d}{d\chi_m} \left[
      \chi_m + \left( \frac{1}{\chi_m - \chi_{m-1}}
      + \frac{f'(\chi_m)}{f(\chi_m)} \right)^{-1}
    \right] \\
    &= 1
    - \left( \frac{1}{\chi_m - \chi_{m-1}}
    + \frac{f'(\chi_m)}{f(\chi_m)} \right)^{-2}
    \frac{d}{d\chi_m} \left( \frac{1}{\chi_m - \chi_{m-1}}
    + \frac{f'(\chi_m)}{f(\chi_m)} \right).
  \end{align*}
  Since
  $\big( \frac{1}{\chi_m - \chi_{m-1}}
  + \frac{f'(\chi_m)}{f(\chi_m)} \big)^{-2}$
  is positive, once more
  we wish to show that
  $\frac{d}{d\chi_m} \big( \frac{1}{\chi_m - \chi_{m-1}}
  + \frac{f'(\chi_m)}{f(\chi_m)} \big)$
  is negative.
  By the log-concavity of $f$ again,
  the term $\frac{f'(\chi_m)}{f(\chi_m)}$
  is decreasing as a function of $\chi_m$,
  so its derivative with respect to $\chi_m$ is negative.
  But, the derivative
  \begin{align*}
    \frac{d}{d\chi_m} \frac{1}{\chi_m - \chi_{m-1}}
    &= -\left(\chi_m - \chi_{m-1}\right)^{-2}
    \frac{d}{d\chi_m} (\chi_m - \chi_{m-1}) \\
    &= -\left(\chi_m - \chi_{m-1}\right)^{-2}
    \left(1 - \frac{d\chi_{m-1}}{d\chi_m}\right)\leq 0,
  \end{align*}
  since $\frac{d\chi_{m-1}}{d\chi_m} \leq 1$
  by inductive assumption.  Hence, $\frac{d\chi_{m+1}}{d\chi_m}\geq 1$.

  Finally, to finish, we write
  \[\frac{d\chi_n}{d\chi_m} = \prod_{k=m}^{n-1} \frac{d\chi_{k+1}}{d\chi_k} \geq 1. \qedhere\]
\end{proof}

\begin{lemma}
  \label{cor:inv}
  For $n\geq 1$,
  the function $\chi_{n+1}: (r, b_n) \to (r, \infty)$ is invertible.
\end{lemma}

\begin{proof}
  The range of $\chi_{n+1}$ is well specified as $\lim_{x \to r} \chi_{n+1}(x) = r$,
  $\lim_{x \to b_n} \chi_{n+1}(x) = \infty$,
  and $\chi_{n+1}$ is continuous on $(r, b_n)$ by Proposition \ref{prop:summary}.
  Since the derivative
  $\frac{d\chi_{n+1}}{dx}
  = \frac{d\chi_{n+1}}{d\chi_1}
  \geq 1$, we conclude that $\chi_{n+1}$ is strictly increasing and therefore invertible.
\end{proof}

\subsection{Existence of MIW sequences
for higher energy functions when $\ell\geq 1$}\label{sec:exist}
Having found an MIW sequence with respect to the Normal density, $\ell=0$, in Corollary \ref{cor:Normal-exist}, we 
fix $f(x)$ as a higher energy function of order $\ell\geq 1$ in this section.
Let $r_0 = -\infty$,
and $r_{\ell+1} = \infty$, where $\ell$ is the number of roots of $f(x)$.
For $1 \leq k \leq \ell$, denote by $r_k$ the $k$th root of $f(x)$.
For each $0 \leq k \leq \ell$, choose a number $N_k\geq 1$ of points that lie in the region $(r_k, r_{k+1})$.

We will define sequences $(\chi_n(x))$ inductively
on each strictly positive region of $f$.  Here, $x$ is a value in each region.
We will denote by $(\chi_n^0(x))$ the sequence on the leftmost region,
and by $(\chi_n^k(x))$ the sequence to the right of the $k$th root.
Denote as well by $a_k(x)$ the choice of $a(x)$ for our $k$th sequence, that we will provide.

\medskip
\noindent
{\bf Base case.} Define $\chi_n^0(x)$ by $\chi_n(x)$, as in Section \ref{sect3.1} with $r=r_0$, $b_0 = r_1$ and $b_{n-1} =b^0_{n-1}$,
on the subinterval $(r_0, b_{n-1}^0)$ of $(r_0, r_1)$.
We want to satisfy the left boundary condition at $-\infty$,
so we choose $a_0(x) \equiv -\infty$.

\medskip
\noindent
{\bf Induction.}
Suppose we have defined $\chi_n^k(x)$ on $(r_k, b^k_{N_k -1})\subset (r_k, r_{k+1})$ for all $0 \leq k \leq K$.
We wish for the left boundary condition of the next sequence
to be the last element of the current one, that is, $a_{K+1}(y) = \chi_{N_K}^K(x)$,
for some $x\in (r_K, r_{K+1})$, $y\in (r_{K+1}, r_{K+2})$ to be determined.
As well, we wish for the first element of the next sequence
to be the right boundary condition of the current one,
$y=\chi_1^{K+1}(y) = \chi_{N_K+1}^K(x)$.
In other words, informally,
\[
  (\chi_{N_K}^K)^{-1} (a_{K+1} (y)) = x
  \qquad\text{and}\qquad
  (\chi_{N_K+1}^K)^{-1} (y) = x,
\]
or
\[
  (\chi_{N_K}^K)^{-1} (a_{K+1} (y)) = (\chi_{N_K+1}^K)^{-1} (y).
\]

Thus, for $y \in (r_{K+1}, r_{K+2})$, we define
\[
  a_{K+1} (y) = \chi_{N_K}^K ((\chi_{N_K+1}^K)^{-1} (y)).
\]

\begin{lemma}
  \label{lem:a-def}
  The definition of $a_{K+1}$ is well-defined.  Moreover, $a_{K+1}(y)<y$ for $y\in (r_{K+1}, r_{K+2})$.
\end{lemma}

\begin{proof}
  Notice, by Lemma \ref{cor:inv}, for $y\in (r_{K+1}, r_{K+2})\subset (r_K, \infty)$ that
  $x=(\chi^K_{N_k +1})^{-1}(y) \in (r_K, b^K_{N_k})\subset (r_K, b^K_{N_k -1})$.  Hence, 
  $x$ belongs to the domain of $\chi^K_{N_k}$, $a_{K+1}(y) = \chi^K_{N_k}(x)\in (r_K, \infty)$ is well-defined, and $(\chi^K_{N_k})^{-1}(a_{K+1}(y))=x$.

  Moreover, as $x\in (r_K, b^K_{N_k})$, by Lemma \ref{lem:chi-bdd}, we have $a_{K+1}(y) = \chi^K_{N_k}(x)< b^K_0=r_{K+1}<y$.
\end{proof}

With this definition of $a_{K+1}$, it is ensured that if we take $x = (\chi_{N_K+1}^K)^{-1}(y)$,
the concatenated sequence
\[
  \left(\chi_1^K(x), \ldots, \chi_{N_K}^K(x),
  \chi_1^{K+1}(y), \ldots, \chi_{N_{K+1}}^{K+1}(y)\right)
\]
will be an MIW sequence.

We now complete the proof of the necessary properties of $a_{K+1}$ to finish the induction step, namely $0\leq a'_{K+1}(y)\leq 1$ for $y\in (r_{K+1}, r_{K+2})$, as we have already shown in Lemma \ref{lem:a-def} that
$a_{K+1}(x)<x$ for $x\in (r_{K+1}, r_{K+2})$. 

\begin{lemma}
  The bound $0 \leq a_{K+1}'(x) \leq 1$ holds for $x\in (r_{K+1}, r_{K+2})$.
\end{lemma}

\begin{proof}  By Lemma \ref{lem:lcderivative}, for $x\in (r_{K+1}, r_{K+2})$, $(\chi^K_{N_k +1})^{-1}(x)\in (r_K, b^K_{N_k})$ is differentiable.  Also, for $y\in (r_K, b^K_{N_k})\subset (r_K, b^K_{N_K-1})$, $\chi^K_{N_K}(y)$ is differentiable.  We compute the derivative
  \begin{align*}
    a_{K+1}'(x)
    &= (\chi_{N_K}^K)'((\chi_{N_K+1}^K)^{-1}(x))
    \ ((\chi_{N_K+1}^K)^{-1})'(x) \\
    &= \frac{(\chi_{N_K}^K)'((\chi_{N_K+1}^K)^{-1}(x))}
    {(\chi_{N_K+1}^K)'((\chi_{N_K+1}^K)^{-1}(x))} \ = \ 
    \frac{(\chi_{N_K}^K)'(y)}{(\chi_{N_K+1}^K)'(y)},
  \end{align*}
  where $y = (\chi_{N_K+1}^K)^{-1}(x)$.
  By Lemma~\ref{lem:lcderivative},
  we have $\frac{d\chi_{N_K+1}^K}{d\chi_{N_K}^K} \geq 1$.
  Thus, $\frac{1}{a_{K+1}'(x)}
  = \frac{(\chi_{N_K+1}^K)'(y)}{(\chi_{N_K}^K)'(y)}
  \geq 1$ as well.  Hence, both derivative bounds hold.
\end{proof}

We now define $\chi_n^{K+1}(x)$ for
$x\in (r_{K+1}, b_{n-1}^{K+1})\subset (r_{K+1}, r_{K+2})$
as $\chi_n(x)$ with the choice $r = r_{K+1}$, $b_0 = r_{K+2}$, and
$a_{K+1}(x) = \chi_{N_K}^K ((\chi_{N_K+1}^K)^{-1} (x))$ to finish the induction step.
\medskip

At this point, we have constructed an MIW sequence bridging the roots of $f$.  What is left to choose is the value $x=x^K$ for each $\chi^K_n(x)$ for $1\leq K\leq \ell$. These choices are limited by the concatenation relation.  So, once a choice is made, say in the right-most interval, this choice is propagated backwards, determining the other choices all the way to the first interval.

As our first parameter, we choose $x^\ell = b_{N_\ell}^\ell$
to satisfy the right boundary condition at $\infty$, according to Lemma \ref{lem:left-bc}.
Then, $x^{\ell-1} = (\chi_{N_{\ell-1}+1}^{\ell-1})^{-1}(x^\ell)$ is determined,
so that the concatenation of the last two sequences will be an MIW sequence.
Continuing recursively yields
$x^{k-1} = (\chi_{N_{k-1}+1}^{k-1})^{-1}(x^k)$,
for $1 \leq k \leq \ell$.
The concatenation
\[
  (\chi_1^0(x_1), \dots, \chi_{N_1}^0(x_1), \dots,
  \chi_1^k(x_k), \dots, \chi_{N_k}^k(x_k), \dots,
  \chi_1^\ell(x_\ell), \dots, \chi_{N_\ell}^\ell(x_\ell))
\]
is then an MIW sequence of $f(x)$,
with the desired number of points in each region,
that satisfies both left and right boundary conditions at $\pm\infty$.
We summarize this in the following proposition.

\begin{proposition}\label{prop:exist}
  There is an MIW sequence of $f$ that satisfies
  the left boundary condition at $-\infty$,
  and the right boundary condition at $\infty$,
  with $N_k\geq 1$ points that lie in the regions $R_k$
  for $0 \leq k \leq \ell$.
\end{proposition}

\subsection{Uniqueness of MIW sequences
for higher energy functions when $\ell\geq 0$}
Let now  $f$ be a higher energy function with order $\ell\geq 0$.
When $f$ has zeros, label them as
$-\infty < r_1 < \cdots < r_\ell < \infty$.
Denote by $r_0 = -\infty$ and $r_{\ell+1} = \infty$.

The MIW sequence $(x_n)_{n=1}^N$ constructed in Proposition~\ref{prop:exist} is in fact unique.
Informally, we can see this from the uniqueness
of the choices made during its construction.
The choices of $a_k(x)$ are necessary to satisfy the recursion relation.
The choice of $x_\ell$ is necessary to satisfy the right boundary condition.
The choices of $x_k$ for $0 \leq k < \ell$
are determined to satisfy the recursion relation.

We now proceed to show this more formally, using a different idea.
We show that if the first elements of two sequences differ,
then one of the two sequences does not satisfy the right boundary condition.  Also, if the first elements are the same, then the two sequences are the same.

Note, when $\ell=0$ and $N=1$, as commented in the introduction, there is no one point MIW sequence $(x_1)$ satisfying \eqref{eq:hdw-seq} with respect to left and right boundary conditions at $-\infty$ and $\infty$.

\begin{lemma}\label{lem:ldineq}
  Consider two $N$-element MIW sequences,
  $(x_n)_{n=1}^N$ and $(y_n)_{n=1}^N$ satisfying the left $-\infty$ boundary condition, and $x_{N+1}$, $y_{N+1}$ right boundary conditions respectively.
  Suppose for every $0 \leq k \leq \ell$,
  that both sequences have the same number of points $N_k\geq 1$
  lying in the intervals $(r_k, r_{k+1})$ for $0\leq k\leq \ell$.

  If $x_1 < y_1$, then $x_{n+1} - x_n < y_{n+1} - y_n$
  for all $1 \leq n < N$.
  Consequentially $x_n < y_n$ for all $1 \leq n \leq N$.

  On the other hand, if $x_1=y_1$, then $x_n=y_n$ for $1\leq n\leq N$.
\end{lemma}

\begin{proof} When $N=1$, the statement already holds.

  Suppose $N\geq 2$.  
  Because each sequence has the same number of points
  in each strictly positive region of $f$,
  the points $x_n$ and $y_n$ lie in the same region
  for every $1 \leq n \leq N$.
  As $x_1 < y_1$ lie in the same region and $f$ is log-concave, we have
  $\frac{f'(y_1)}{f(y_1)} < \frac{f'(x_1)}{f(x_1)}$ (Lemma \ref{lem:higher-order}).
  Then, with respect to the left boundary conditions,
  \begin{align}
    \label{x1 x2}
    \frac{1}{x_2 - x_1}
    = \frac{f'(x_1)}{f(x_1)}
    > \frac{f'(y_1)}{f(y_1)}
    = \frac{1}{y_2 - y_1}.
  \end{align}
  Thus, we can bound $x_2 - x_1 < y_2 - y_1$.

  Suppose that $x_{n} < y_{n}$
  and $x_{n+1} - x_{n} < y_{n+1} - y_{n}$
  for some $1 \leq n < N$.
  This implies that $x_{n+1} < y_{n+1}$.
  Because $x_{n+1}$ and $y_{n+1}$ are in the same region, as before,
  $\frac{f'(y_{n+1})}{f(y_{n+1})} < \frac{f'(x_{n+1})}{f(x_{n+1})}$.
  Then,
  \begin{align*}
    \frac{1}{x_{n+2} - x_{n+1}}
    = \frac{1}{x_{n+1} - x_{n}} + \frac{f'(x_{n+1})}{f(x_{n+1})}
    > \frac{1}{y_{n+1} - y_{n}} + \frac{f'(y_{n+1})}{f(y_{n+1})}
    = \frac{1}{y_{n+2} - y_{n+1}}.
  \end{align*}
  Thus, we can bound $x_{n+2} - x_{n+1} < y_{n+2} - y_{n+1}$.  Hence, by this iteration, $x_n<y_n$ for $1\leq n\leq N$.

  The last statement follows the same argument:  In \eqref{x1 x2}, if $x_1=y_1$, we would conclude $x_2=y_2$, and so on.
\end{proof}

\begin{proposition}\label{prop:unique}
  There is at most one MIW sequence $(x_n)_{n=1}^N$
  that satisfies the left boundary condition at $-\infty$
  and the right boundary condition at $\infty$, with $N_k\geq 1$ points in the intervals $(r_k, r_{k+1})$ for $0\leq k\leq \ell$ when $\ell\geq 1$, and $N=N_0\geq 2$ when $\ell=0$.
\end{proposition}

\begin{proof}
  Let $(x_n)_{n=1}^N$ and $(y_n)_{n=1}^N$ be two such sequences.  If they are different, by Lemma \ref{lem:ldineq}, $x_1\neq y_1$.  To make a choice, suppose that $x_1 < y_1$.
  By Lemma~\ref{lem:ldineq}, we have the relations
  $x_{n+1} - x_n < y_{n+1} - y_n$
  for all $1 \leq n < N$,
  and thus $x_n < y_n$ for all $1 \leq n \leq N$.
  Our right boundary conditions require that
  \begin{align*}
    -\frac{1}{x_N - x_{N-1}} = \frac{f'}{f}(x_N)
    \qquad\text{and}\qquad
    -\frac{1}{y_N - y_{N-1}} = \frac{f'}{f}(y_N).
  \end{align*}
  However, by log-concavity of $f$ (Lemma \ref{lem:higher-order}),
  $\frac{f'(x_N)}{f(x_N)} > \frac{f'(y_N)}{f(y_N)}$ and so
  $-\frac{1}{x_N - x_{N-1}} < -\frac{1}{y_N - y_{N-1}}$,
  or in other words $y_N - y_{N-1}< x_N - x_{N-1}$.  This contradicts that $x_1<y_1$.
\end{proof}

\begin{proof}[\textup{\textbf{Proof of Theorem~\ref{thm:exist-unique}}}]
  Existence follows from Corollary \ref{cor:Normal-exist} for $\ell=0$, and from Proposition~\ref{prop:exist} for $\ell\geq 1$.  Uniqueness follows from Proposition~\ref{prop:unique}.
\end{proof}

\section{Proof of Theorem \ref{thm:gaps}:  Gap sizes and span of MIW sequences}
\label{gap_sect}

Suppose $f(x)$ is a higher energy function with order $\ell\geq 0$.  When $f$ has roots, label them as
$r_1 < \cdots < r_\ell$.  Let $r_0=-\infty$ and $r_{\ell+1}=\infty$
By Theorem \ref{thm:exist-unique}, let $(x_n)_{n=1}^N$ be
the $N$-element MIW sequence of $f$
satisfying the left boundary condition at $-\infty$,
and the right boundary condition at $\infty$,
such that in regions
$(r_k, r_{k+1})$ there are
$N_k\uparrow\infty$ points for $0\leq k\leq \ell$.  Note, when $\ell=0$, the sequence $(x_n)_{n=1}^N$ is symmetric by Corollary \ref{cor:Normalsymmetry} with $\lfloor N/2\rfloor$ points below and above the origin.

Recall the definition of $n(t)$ in \eqref{n(t) eq} as the index of the closest point in $(x_n)_{n=1}^N$ to the left of $t\in \R$.  Recall also the MIW relation \eqref{eq:hdw-seq}.

The strategy of proof of Theorem \ref{thm:gaps}, at the end of the section, is to leverage the higher energy form of $f$, with $\ell+1$ local maxima, which indicates that $x_2 - x_1$, $x_{n(r)+1} - x_{n(r)}$ for roots $r$, and $x_N - x_{N-1}$
are the local maxima of $x_{n+1} - x_n$ over $1\leq n\leq N-1$.  Some development is needed to show these `boundary' gaps vanish, as well as estimates for the orders of $x_N, x_1$ along the way in the following lemmas.

\begin{lemma}
  \label{lem:4.1}
  Let $f$ be 
  strictly positive on the finite interval $(u, v)$ where $u$ and $v$ are consecutive zeros of $f$.
  Let $(x_n)_{n=1}^N$ be the part of an MIW sequence of $f$ contained in $(u,v)$.
  Then, $\lim_{N \to \infty} x_1 \to u$
  and $\lim_{N \to \infty} x_N \to v$.
\end{lemma}

\begin{proof}
  Let $\varepsilon<0$ be such that
  $\frac{1}{x_1 - u} \leq \frac 1\varepsilon$, that is
  $x_1 \geq u + \varepsilon$.
  Note, as the left boundary point $x_0\leq u$ that $\frac{1}{x_1-x_0}\leq \frac{1}{\varepsilon}$.  Moreover, by the form of $f$ (cf. Lemma \ref{lem:3.5}), we can take $\varepsilon$ smaller if necessary so that also $f'(u+\varepsilon)>0$.
  Also, by log-concavity of $f$ (Lemma \ref{lem:higher-order}), $\frac{f'}{f}$ is decreasing on $(u, v)$.

  We can bound, as the MIW sequence $(x_n)_{n=1}^N$ is strictly increasing,
  \begin{align}
    \label{lem4.1help}
    0<\frac{1}{x_{n+1} - x_n}
    = \frac{1}{x_1 - x_0} + \sum_{k=1}^n \frac{f'(x_k)}{f(x_k)}
    \leq \frac 1\varepsilon + n\frac{f'(u + \varepsilon)}{f(u + \varepsilon)}.
  \end{align}
  Then, by the MIW relation,
  \begin{align*}
    v-u\geq  x_{N} - x_1
    = \sum_{k=1}^{N-1} x_{k+1} - x_k
    \geq \sum_{k=1}^N
    \frac{1}{\frac 1\varepsilon + k\frac{f'(u + \varepsilon)}{f(u + \varepsilon)}}.
  \end{align*}
  The right hand side is diverges to $\infty$ as $N \to \infty$, yielding a contradiction.
  Hence, no such $\varepsilon$ can exist.  Therefore, for every $\varepsilon > 0$,
  for $N$ sufficiently large,
  $u <x_1 < u + \varepsilon$.

  The other limit holds by the same argument applied with the function $f(-x)$.
\end{proof}

We now state an estimate of $x_N$, subject to `spanning' assumption.
\begin{lemma}\label{lem:xNbound}
  Let $f$ be 
  strictly positive on $(u, \infty)$ for $u=r_\ell$ when $\ell\geq 1$ and $u=0$ when $\ell=0$. 
  Let $(x_n)_{n=1}^N$ denote
  the part of an MIW sequence of $f$ contained in $(u, \infty)$.

  Suppose for some $v>u$, such that $p'_\ell(v)/p_\ell(v) \leq v/4$, the number of points beyond $v$, that is $N-n(v)$, diverges.  Then,
  we have
  $x_N^2 \leq  x_{n(v)+1}^2+ 4(1+\log(N - n(v)))$.
\end{lemma}

\begin{proof}
  As $\frac{p'_\ell(x)}{p_\ell(x)} = O(1/x)$, let $v$ be such that
  $\frac 12 x \leq x - 2\frac{p'_\ell(x)}{p_\ell(x)}$ for $x\geq v$.
  Then, when $j\geq n(v)$,
  \[
    0<\frac 12 \sum_{k=j+1}^N x_k
    \leq \sum_{k=j+1}^N x_k - 2\frac{p_\ell'(x_k)}{p_\ell(x_k)}.
  \]
  Inverting both sides, noting the MIW sequence $(x_n)_{n=1}^N$ is increasing,
  \begin{equation}
    \label{lem4.2help}
    \frac{1}{\sum_{k=j+1}^N x_k - 2\frac{p_\ell'(x_k)}{p_\ell(x_k)}}
    \leq \frac{2}{\sum_{k=j+1}^N x_k}
    \leq \frac{2}{(N-j) x_{j+1}}.
  \end{equation}
  Note as $-f'(x)/f(x) = x -2p'_\ell(x)/p_\ell(x)$, by the MIW relation \eqref{MIW-relation1}, the fact $x_j<x_{j+1}$, and the last display \eqref{lem4.2help} that
  \begin{align*}
    x_N^2
    &= x_{n(v)+1}^2 + \sum_{j={n(v)+1}}^{N-1} x_{j+1}^2 - x_j^2
    \leq x_{n(v)+1}^2 + 2 \sum_{j={n(v)+1}}^{N-1} x_{j+1}(x_{j+1} - x_j) \\
    &= x_{n(v)+1}^2 + 2 \sum_{j={n(v)+1}}^{N-1}
    \frac{x_{j+1}}{\sum_{k=j+1}^N x_k - 2\frac{p_\ell'(x_k)}{p_\ell(x_k)}}
    \leq x_{n(v)+1}^2 + 4 \sum_{j={n(v)+1}}^{N-1} \frac{1}{N-j}.
  \end{align*}
  Hence,
  $x_N^2 \leq x_{n(v)+1}^2
  + 4(1+\log(N - n(v)))$, as desired.
\end{proof}

\begin{lemma}
  \label{lem:4.3}
  Let $\ell\geq 1$ and note $f$ is 
  strictly positive on $(r_\ell, \infty)$.
  Let $(x_n)_{n=1}^N$ be the part of an MIW sequence of $f$ contained in $(r_\ell,\infty)$.
  Then, the limit $\lim_{N \to \infty} x_1 = r_\ell$ holds.

  Similarly, for $\ell\geq 1$, note $f$ is strictly positive on $(-\infty, r_1)$.
  Let $(x_n)_{n=1}^N$ be the part of an MIW sequence contained in $(-\infty, r_1)$.  Then, $\lim_{N\to \infty} x_N=r_1$.
\end{lemma}

\begin{proof}
  Let $u=r_\ell$.  Suppose $\varepsilon > 0$ is such that $\frac{1}{x_1 - u} \leq \frac{1}{\varepsilon}$, that is $x_1\geq \varepsilon + u$. Hence, $x_j>x_1\geq \varepsilon + u$ for $j\geq 1$.  We can arrange by the form of $f$ that also $f'(u + \varepsilon)>0$ (cf. Lemma \ref{lem:3.5}).  Following the proof of Lemma \ref{lem:4.1}, noting \eqref{lem4.1help},
  we have that
  \begin{align}
  \label{lem4.3help}
    x_n
    &\geq x_j +\sum_{k=j}^{n-1} \frac{1}{\frac{1}{\varepsilon} + k\frac{f'(u+ \varepsilon)}{f(u+\varepsilon)}} \geq u + \sum_{k=j}^{n-1} \frac{1}{\frac{1}{\varepsilon} + k\frac{f'(u+ \varepsilon)}{f(u+\varepsilon)}}.
  \end{align}
  Hence, $x_n$ must grow at least logarithmically with $n$.  In particular, $x_{\lfloor \sqrt{N}\rfloor} >v$ for all large $N$, with respect to any $v>u$.  Hence, $N- n(v) \geq N-\sqrt{N}$ diverges, and by Lemma \ref{lem:xNbound} applied to the $v$ specified there, we have
  $x^2_N \leq x_{n(v)+1}^2 + 4(1+\log N)$.  But, from the first inequality in \eqref{lem4.3help} with $n=N$ and $j =n(v)+1$, we have $x_N\geq x_{n(v)+1} + c\log(N/\sqrt{N})$ for a constant $c>0$. Then, $x^2_N \geq (x_{n(v)+1} + (c/2)\log N)^2 \geq x^2_{n(v)+1} + ((c/2)\log N)^2$, yielding a contradiction in the growth of $x^2_N - x^2_{n(v)+1}$.
 Hence, $r_\ell=u<x_1\leq u+\varepsilon$, and so $x_1\rightarrow u$ as $N\uparrow\infty$.

  The last statement follows the same argument by considering $f(u-x)$ for $u=r_1$.
\end{proof}

We now give a lower bound for $x_N$. 
\begin{lemma}
  \label{prop:gaps-xN}
  Let $f$ be 
  strictly positive on $(u,\infty)$ where
  $u=r_\ell$ if $\ell\geq 1$ and $u=0$ when $\ell=0$.
  Let $(x_n)_{n=1}^N$ denote the part of
  an MIW sequence of $f$ contained in $(u,\infty)$.  
  Then,
  the term $x_N > \sqrt{\log N}$ as $N\uparrow\infty$.
\end{lemma}

\begin{proof}
  For $x>u$, we have $p_\ell'(x)/p_\ell(x) = \sum_{j=1}^\ell (x-r_j)^{-1}>0$ when $\ell\geq 1$.  When $\ell=0$, $p'_\ell(x) = 0$.  In both cases,
  $x - 2\frac{p_\ell'(x)}{p_\ell(x)} \leq x$.

  Then, by \eqref{MIW-relation1} and as the MIW sequence $(x_n)_{n=1}^N$ is increasing,
  \[
    0 < \frac{1}{x_{n+1} - x_n}
    = \sum_{k=n+1}^N x_k - 2\frac{p_\ell'(x_k)}{p_\ell(x_k)}
    \leq  \sum_{k=n+1}^N x_k
    \leq (N-n) x_N.
  \]
  That is to say,
  $\frac{1}{(N-n) x_N}
  \leq x_{n+1} - x_n$ for $1\leq n\leq N-1$.

  Then, as $x_1\geq 0$,
  \begin{align*}
    x_N
    &= x_{1} + \sum_{k=1}^{N-1} x_{k+1} - x_k
    \geq \sum_{k=1}^{N-1} \frac{1}{(N-k)x_N}
    = \frac{1}{x_N} \sum_{k=1}^{N-1} \frac 1k.
  \end{align*}
  Thus, we can bound
  $x_N^2
    \geq \sum_{k=1}^{N-1} \frac 1k
    \geq \log N$.
\end{proof}

\begin{lemma}\label{lem:tail-gap}
  Let $f$ be strictly positive in $(u, \infty)$ where $u=r_\ell$ when $\ell\geq 1$ and $u=0$ when $\ell=0$.
  Let $(x_n)_{n=1}^N$ denote the part of 
  an MIW sequence of $f$ contained in $(u,\infty)$.
  Then, for all $N$ sufficiently large,
  $x_{N} - x_{N-1} < \frac{2}{\sqrt{\log N}}$.

  Similarly, with respect to the part of the MIW sequence contained in $(-\infty, v)$ where $f$ is strictly positive and $v=r_1$ when $\ell\geq 1$ and $v=0$ when $\ell=0$, we have that $x_2-x_1 < \frac{2}{\sqrt{\log N}}$. 
\end{lemma}

\begin{proof}
  Observe that $-\frac{f'(x)}{f(x)}=x - 2\frac{p_\ell'(x)}{p_\ell(x)} = x + \calO(1/x)$ as $x\uparrow\infty$.
  Since $x_N\geq \sqrt{\log N}$ by Lemma \ref{prop:gaps-xN}, we have
  $x_N - 2\frac{p_\ell'(x_N)}{p_\ell(x_N)} \geq \frac 12 x_N$ for large $N$.
  Thus, from the right $\infty$ boundary condition,
  \begin{align*}
    x_{N} - x_{N-1}
    = \frac{1}{x_N - 2\frac{p_\ell'(x_N)}{p_\ell(x_N)}}
    \leq \frac{2}{x_N}
    < \frac{2}{\sqrt{\log N}}.
  \end{align*}

  The argument for $x_2-x_1$ on the left end, follows by symmetry of $f$.
\end{proof}

\begin{lemma}\label{lem:gaps}
  For each $t \in \R$,
  the difference $x_{n(t)+1} - x_{n(t)} \to 0$.
  Therefore, $\lim_{N \to \infty} x_{n(t)+1} = t$.
\end{lemma}

\begin{proof}
  Let $m$ be a local maximum of $f$.  Let $r$ be the closest zero of $f$ to the right of $m$, or if there is no such zero, let $r=\infty$.  Similarly, let $s$ be the closest zero to the left of $m$ or if there is no such zero, let $s=-\infty$.

  On the intervals $(m,r)$ and $(s,m)$, the function $\frac{f'(x)}{f(x)}$ is negative and positive respectively.
  Thus, given $x_k \in (m,r)$, by the MIW relation,
  $\frac{1}{x_{k+1} - x_k} - \frac{1}{x_k - x_{k-1}} = \frac{f'(x_k)}{f(x_k)}< 0$,
  or $x_k - x_{k-1} < x_{k+1} - x_k$.
  Likewise, given $x_k \in (s,m)$, we can bound
  $x_k - x_{k-1} > x_{k+1} - x_k$.

  Hence, the
  interval sizes of $(x_n)_{n=1}^N$ are bounded by $x_{n(r)+1} - x_{n(r)}$ at roots $r$ of $f$, or by $x_2-x_1$ and $x_N-x_{N-1}$ at the extremities.

  By Lemmas \ref{lem:4.1} and \ref{lem:4.3},
  $x_{n(r)+1} - x_{n(r)}\rightarrow 0$.  By Lemma \ref{lem:tail-gap}, applied with $u=r_\ell$, $v=r_1$ when $\ell\geq 1$ and with $u=v=0$ when $\ell=0$,
  $x_N-x_{N-1}, x_2-x_1\rightarrow 0$.  Hence, for all $\ell\geq 0$, the gaps $x_{n(t)+1} - x_{n(t)} \rightarrow 0$ for $t\in \R$.
\end{proof}

Finally, we give an upper order to the points $|x_1|$ and $x_N$.
\begin{lemma}
  \label{lem:upperorder}
  We have $x_N, |x_1| = O(\sqrt{\log N})$.
\end{lemma}

\begin{proof}
  By Lemma \ref{lem:gaps}, the gaps $x_{n(t)+1} - x_{n(t)} \rightarrow 0$ for $t\in \R$.  Also, by Lemma \ref{prop:gaps-xN}, 
  we have $x_N\geq \sqrt{\log N}$.  Hence, as a consequence, the number of points in the sequence beyond any $v>r_\ell$ when $\ell\geq 1$ or $v>0$ when $\ell=0$ diverges as $N\uparrow\infty$.  We may now apply Lemma \ref{lem:xNbound} to see that $x^2_N \leq x^2_{n(v)+1} + \calO(\log N)$, and by Lemma \ref{lem:gaps} to conclude $x^2_N = \calO(1) + \calO(\log N)$, and so $x_N = O(\sqrt{\log N})$. 

  The same bound on $|x_1|$ holds by symmetry of $f$.
\end{proof}

\medskip
\noindent
{\bf Proof of Theorem \ref{thm:gaps}.}  The vanishing gap property is given in Lemma \ref{lem:gaps}.  Lower and upper bounds of the extremities $|x_1|, x_N$ are given in Lemmas \ref{prop:gaps-xN} and \ref{lem:upperorder}.
\qed

\section{Proof of Theorem \ref{thm:converge}: Convergence in Wasserstein-$1$ distance}
\label{converge-sect}

Let $f$ be a higher energy function of order $\ell\geq 0$.  Consider an interval $(a,b)$ where $f$ is strictly positive such that $a$ is a zero of $f$ or $a=-\infty$ and $b$ is the next zero of $f$ to the right of $a$ or $b=\infty$.
We will define $P=P_{a,b}$ as the distribution on $(a,b)$ with density proportional to $f$.  Consider the part of the MIW sequence constructed in Theorem \ref{thm:exist-unique} contained in $(a,b)$, with $N=N_{a,b} \geq 2$ points, and denote it as $(x_n)_{n=1}^N$.  Let $Q=Q_{a,b}$ be the empirical measure associated to this sequence.

Recall $n(t)$ defined in \eqref{n(t) eq}.  To bound $|\Exp_P[h]-\Exp_Q[h]|$,
as in \cite{mckeague_levin_2016, chen_thanh_2020} with respect to the Normal density, we will compare $Q$ with respect to another `intermediate' continuous distribution $R=R_{a,b}$ with density
\begin{align*}
  \rho(t) = \frac{1}{(N-1)(x_{n(t)+1} - x_{n(t)})},
\end{align*}
supported on $(x_1, x_N)\subset (a,b)$.

Then, to compare $R$ with $P$, 
we will use the `differential equation' method with respect to a Stein equation.
Define 
in terms of a differentiable $1$-Lipschitz function $h\in \mathcal{L}$, the function
\begin{align}
  \label{gh-eqn}
  g_h(x) = \frac{1}{f(x)} \int_a^x f(t)(h(t) - \Exp_P[h])\ dt.
\end{align}
Note that $g_h$ satisfies the `Stein equation',
\begin{align}\label{Steineqn}
  g_h'(x) + \frac{f'(x)}{f(x)}g_h(x) = h(x) - \Exp_P[h].
\end{align}
Such `Stein equations' and subsequent approximations have been considered in wide generality; see \cite{ernst_swan_2022}, \cite{Gaunt}.

We will bound $|\Exp_R[h] - \Exp_Q[h]|$ in Section \ref{coupling bound sect}
and 
$|\Exp_R[h] - \Exp_P[h]|$
in terms of $g_h$ in Section \ref{converge-sect2}.
Both will be shown, after development, of order
$\calO(\frac{x_N - x_1}{N})$.

On finite intervals $\frac{x_N - x_1}{N} = \calO(1/N)$, and on infinite intervals,
by Theorem \ref{thm:gaps}, as $x_N, x_1 = \calO(\sqrt{\log N})$, we have $\frac{x_N - x_1}{N} = \calO(\sqrt{\log N}/N)$.  
We will piece together now these estimates over the intervals $(r_k, r_{k+1})$ between any zeros of $f$ and conclude the proof of Theorem \ref{thm:converge} in Section \ref{proofconvergesect} that the Wasserstein-$1$ distance
$d(P, Q) = \sup_{h\in \mathcal{L}} |\Exp_P[h] - \Exp_Q[h]| \leq \calO(\frac{\sqrt{\log N}}{N})$, 
as desired.

It will be of later use to express the second derivative
\begin{align}
  g_h''(x)
    &= 2\left(\frac{f'(x)}{f(x)}\right)^2 g_h(x)
    - \frac{f''(x)}{f(x)}g_h(x) - \frac{f'(x)}{f(x)} (h(x) - \Exp_P[h]) + h'(x).
    \label{eq:gh-split}
\end{align}
We also state a general estimate which will be helpful in the sequel, especially in places to bound terms uniformly over $h\in \mathcal{L}$.

\begin{lemma}\label{lem:bd-h}
  Let $P$ be a probability distribution on $(a,b)$,
  where $-\infty \leq a < b \leq \infty$.
  Let $h$ be a $1$-Lipschitz function,
  and $y \in (a,b)$,
  Let also $X$ be a random variable with distribution $P$.

  For $y\in \R$, we have
  $|h(y) - \Exp_P[h]| \leq \Exp_P |y-X| \leq |y| + \Exp_P[|X|]$.

  If $a$ is finite, then
  $
  |h(y) - \Exp_P[h]| \leq \Exp_P[X] + y - 2a$.

  If $b$ is finite, then
  $
  |h(y) - \Exp_P[h]| \leq 2b - \Exp_P[X] - y$.

  If $a$, $b$ are both finite, then
  $
  |h(y) - \Exp_P[h]| \leq b - a$
\end{lemma}

\begin{proof}
  For $y\in (a,b)$,
  we compute, as $h\in \mathcal{L}$, that
  $
  |h(y) - \Exp_P[h]|
  \leq \Exp_P |y - X|\leq |y| + \Exp_P[|X|]$.
  Moreover, if $a$ is finite, then as $|y-X| \leq y-a + a-X$, we have
  $|h(y) - \Exp_P[h]|  \leq
  = \Exp_P[X]+ y - 2a$.
  Likewise, when $b$ is finite, $|h(y) - \Exp_P[h]| \leq 2b - \Exp_P[X] - y$.
  Summing these, we obtain the last inequality, $2|h(y) - \Exp_P[h]| \leq 2(b-a)$.
\end{proof}

\subsection{Bound on $|\Exp_R[h] - \Exp_Q[h]|$}
\label{coupling bound sect}
Recall the distributions $Q$ and $R$ defined in terms of the MIW sequence $(x_n)_{n=1}^N$ of a higher energy function $f$.
The following, in the case of the Normal density $f$, was started in \cite{mckeague_levin_2016, chen_thanh_2020}.  The argument is the same, given here for the reader's convenience.

\begin{proposition}\label{prop:wass1}
  We have that
  the Wasserstein-$1$ distance
  $\sup_{h \in \calL} |\Exp_R[h] - \Exp_Q[h]| \leq \frac{x_N - x_1}{N-1}$.
\end{proposition}

\begin{proof}
  Let $X$ be a random variable with distribution
  $Q = \frac 1N \sum_{n=1}^N \delta_{x_n}$,
  and $Y$ be a random variable with distribution $R$.  We couple them as follows.

  Choose a sequence $(z_n)_{n=1}^{N+1}$
  such that $z_1 = x_1$, $z_{N+1} = x_N$,
  and $\int_{z_n}^{z_{n+1}} \rho(t)\ dt = \frac 1N$.
  Since $\int_{x_k}^{x_{k+1}} \rho(t)\ dt = \frac{1}{N-1}$,
  we can bound $x_n \leq z_{n+1} \leq x_{n+1}$.
  Given $Y \sim R$,
  choose $X = x_{n^*(Y)}$,
  where $n^*(t) = \#\{z_k \leq t\}$.
  Then, $X \sim Q$,
  and $|Y - X| \leq x_{n^*(Y) + 1} - x_{n^*(Y)}$.

  Then, in terms of the coupled expectation $\Exp_{Q, R}$, 
  \[
    \sup_{h \in \calL} |\Exp_R[h] - \Exp_Q[h]|
    \leq \Exp_{Q,R} |Y - X|
    \leq \sum_{n=1}^{N-1} \frac{x_{n+1} - x_n}{N-1}
    = \frac{x_N - x_1}{N-1}. \qedhere
  \]
\end{proof}

\subsection{Bound on $|\Exp_R[h]-\Exp_P[h]|$ in terms of $g_h$}
\label{converge-sect2}

We show a bound between the intermediate distribution $R$
and the continuous distribution $P=P_{a,b}$ restricted to an interval $(a,b)$, in terms of $g_h$.

\begin{proposition}\label{prop:wass}
  Let $f$ be a higher energy function
  strictly positive on the interval $(a,b)$, where $a$ is a root of $f$ or $a=-\infty$, and $b$ is the next root after $a$ or if there is no such root $b=\infty$.
  Let $(x_n)_{n=1}^N$ be part of an MIW sequence of $f$ contained in $(a,b)$.

  Then, for $\beta \in [0,1]$,
  and differentiable 1-Lipschitz $h\in \mathcal{L}$, we have
  \begin{align*}
    &|\Exp_R[h] - \Exp_P[h]|\\
    &\quad \leq \frac{1}{N-1} \Big( \Big|
      \frac{g_h(x_N)}{x_{N} - x_{N-1}}
      - \frac{g_h(x_1)}{x_{2} - x_{1}} + \beta\Big(\frac{f'}{f} g_h\Big)(x_N)
    + (1-\beta)\Big(\frac{f'}{f} g_h\Big)(x_1) \Big| \\
    &\quad \quad \quad + (x_N - x_1)\big(1 + \sup_{x \in (x_1, x_N)} |g_h''(x)|\big) \Big).
    \stepcounter{equation}\tag{\theequation}\label{eq:stein-bound}
  \end{align*}
\end{proposition}

\begin{proof}
  Taking the expectation of each side of Stein's equation \eqref{Steineqn},
  with respect to $R$,
  \begin{align*}
    \Exp_R[h] - \Exp_P[h]
   & =
   \int_a^b \rho(t) \big(g_h' + \frac{f'}{f} g_h\big)\ dt \\
   &= \sum_{n=1}^{N-1} \rho(x_n) \left( \int_{x_n}^{x_{n+1}} \frac{f'}{f} g_h\ dt
   + \int_{x_n}^{x_{n+1}} g_h'\ dt \right) \\
   &= \frac{1}{N-1} \Big( \sum_{n=1}^{N-1} \int_{x_n}^{x_{n+1}}
     \frac{\big(\frac{f'}{f} g_h\big)(t)}{x_{n+1} - x_{n}}\ dt
   + \sum_{n=1}^{N-1} \frac{g_h(x_{n+1}) - g_h(x_{n})}{x_{n+1} - x_{n}} \Big).
  \end{align*}

  By summation by parts, as $\frac{1}{x_{n+1} - x_n} - \frac{1}{x_n - x_{n-1}} = \frac{f'(x_n)}{f(x_n)}$ by the MIW relation \eqref{eq:hdw-seq},
  \begin{align}
   & \sum_{n=1}^{N-1} \frac{g_h(x_{n+1}) - g_h(x_{n})}{x_{n+1} - x_{n}}
   = \frac{g_h(x_N)}{x_{N} - x_{N-1}}
   - \frac{g_h(x_1)}{x_{2} - x_{1}}
   - \sum_{n=2}^{N-1} \left(\frac{f'}{f} g_h\right)(x_n) \nonumber \\
   &\ \ \ = \frac{g_h(x_N)}{x_{N} - x_{N-1}}
   - \frac{g_h(x_1)}{x_{2} - x_{1}}
   + \left(\frac{f'}{f} g_h\right)(x_1)
   - \sum_{n=1}^{N-1} \int_{x_n}^{x_{n+1}}
   \frac{\left(\frac{f'}{f} g_h\right)(x_n)}{x_{n+1} - x_n}\ dt. \label{eq:convex}
  \end{align}
  
   Instead of focusing on the left endpoint in \eqref{eq:convex},
  we can work with the right endpoint:
  \begin{align*}
    \big(\frac{f'}{f} g_h\big)(x_1)
    - \sum_{n=1}^{N-1} \int_{x_n}^{x_{n+1}}
    \frac{\left(\frac{f'}{f} g_h\right)(x_n)}{x_{n+1} - x_n}\ dt
    = \left(\frac{f'}{f} g_h\right)(x_N)
    - \sum_{n=1}^{N-1} \int_{x_n}^{x_{n+1}}
    \frac{\left(\frac{f'}{f} g_h\right)(x_{n+1})}{x_{n+1} - x_n}\ dt.
  \end{align*}

  We may take a convex combination of the these expressions in terms of $\beta\in [0,1]$, and stick back into \eqref{eq:convex}.
 Thus, 
  \begin{align*}
&    (N-1)|\Exp_{R}[h] - \Exp_P[h]|\\
&\leq \big|\frac{g_h(x_N)}{x_{N} - x_{N-1}}
- \frac{g_h(x_1)}{x_{2} - x_{1}}
+\beta \big(\frac{f'}{f} g_h\big)(x_N)+ (1-\beta)\big(\frac{f'}{f} g_h\big)(x_1)\big| \\
&\ \ \ \ + \beta \sum_{n=1}^{N-1} \int_{x_n}^{x_{n+1}} \frac{1}{|{x_{n+1} - x_n}|}
\big|\big(\frac{f'}{f} g_h\big)(t)
- \big(\frac{f'}{f} g_h\big)(x_{n+1})
\big|\ dt \\
&\ \ \ \ + (1-\beta)\sum_{n=1}^{N-1} \int_{x_n}^{x_{n+1}} \frac{1}{|{x_{n+1} - x_n}|}
\big|\big(\frac{f'}{f} g_h\big)(t)
- \big(\frac{f'}{f} g_h\big)(x_n)
\big|\ dt .
\end{align*}

  Define
  \begin{align*}
  K_h:=  \sup_{x, y \in (x_1, x_N)}\frac{1}{|y-x|}
    \big|\big(\frac{f'}{f} g_h\big)(y)
    - \big(\frac{f'}{f} g_h\big)(x)\big|.
  \end{align*}
  We can compute, as $|x_{n+1}-x_n| = x_{n+1}-x_n$ as $(x_n)_{n=1}^N$ is increasing, that
  \begin{align*}
    \sum_{n=1}^{N-1} \sup_{s,t \in [x_n, x_{n+1}]}
    \big|\big(\frac{f'}{f} g_h\big)(t) - \big(\frac{f'}{f} g_h\big)(s)\big|
    \leq \sum_{n=1}^{N-1} K_h (x_{n+1} - x_{n})
    = K_h(x_N - x_1).
  \end{align*}

  Also, noting Stein's equation \eqref{Steineqn}, and $h$ is differentiable and $1$-Lipschitz,
  \begin{align*}
    K_h
    &\leq \sup_{x \in (x_1, x_N)} \big|\big(\frac{f'}{f} g_h\big)'\big|
    = \sup_{x \in (x_1, x_N)} |(h - \Exp_P[h] - g_h')'| \\
    &\leq \sup_{x \in (x_1, x_N)} |h'(x)| + |g_h''(x)|
    \leq 1 + \sup_{x \in (x_1, x_N)} |g_h''(x)|.
  \end{align*}

 Putting these estimates together
  yields \eqref{eq:stein-bound}.
\end{proof}

\subsection{Bounds of 
$\frac{f'(x_1)}{f(x_1)} g_h(x_1)$, $\frac{g_h(x_1)}{x_2-x_1}$ and $g''_h(x)$ near finite zeros of $f$}

Let $f$ be a higher energy function,
strictly positive on $(a,b)$, where $-\infty< a < b \leq \infty$.
Suppose $f(a)=0$, and $h$ is a $1$-Lipschitz differentiable function.

We begin with asymptotics near $x=a$, which will be helpful.
\begin{lemma}\label{lem:6.10}

  We have
  \begin{align*}
    g_h(x) &= \frac{1}{3} (x-a) (h(x) - \Exp_P[h]) + \calO((x-a)^2),\\
    \frac{f'(x)}{f(x)} g_h(x)
           &= \frac{2}{3} (h(x) - \Exp_P[h]) + \calO((x-a)).
  \end{align*}
  Here,
  $\calO((x-a))$ and $\calO((x-a)^2)$ do not depend on $h$.
\end{lemma}

\begin{proof}
  Denote $F(x) = \int_a^x f(t)\ dt$.  
  Integrating by parts,
  \begin{align*}
    \int_a^x f(t) (h(t) - \Exp_P[h])\ dt
    &= F(x)(h(x) - \Exp_P[h]) - F(a)(h(a) - \Exp_P[h])
    - \int_a^x F(t) h'(t) dt \\
    &= F(x)(h(x) - \Exp_P[h]) - \int_a^x F(t) h'(t) dt.
  \end{align*}
  Using $f(x) = \frac{1}{2}f''(a)(x-a)^2 + \calO((x-a)^3)$ for $x$ near $a$ (cf. Lemma~\ref{lem:dbl-zro}), which implies $F(x) = \frac{1}{6}f''(a)(x-a)^3 + \calO((x-a)^4)$, 
  $h'(x) = \calO(1)$, and $h(x) - \Exp_P[h] = \calO(x-2a + \Exp_P[X])= \calO(1)$ (Lemma \ref{lem:bd-h}), we have
  \begin{align*}
    &F(x)(h(x) - \Exp_P[h]) - \int_a^x F(t) h'(t) dt\\
    &\ \ = \frac{f''(0)}{6}(x-a)^{3}(h(x) - \Exp_P[h]) + \calO((x-a)^{4}) + \calO((x-a)^4),
  \end{align*}
  where $\calO(\cdot)$ does not depend on $h$.

  The first estimate now follows, noting again $h(x) - \Exp_P[X] = \calO(1)$ for $x$ near $a$, as
  \begin{align*}
    g_h(x)
    &= \frac{\int_a^x f(t) (h(t) - \Exp_P[h])\ dt}{f(x)}
    = \frac{\frac{f''(a)}{6}(x-a)^{3}(h(x) - \Exp_P[h]) + \calO((x-a)^{4})}
    {\frac{1}{2}f''(a)(x-a)^2 + \calO((x-a)^{3})} \\
    &= \frac{\frac{1}{3}(x-a)(h(x) - \Exp_P[h]) + \calO((x-a)^2)}
    {1 + \calO((x-a))}
    = \frac{1}{3} (x-a) (h(x) - \Exp_P[h]) + \calO((x-a)^2),
  \end{align*}
  where $\calO(\cdot)$ does not depend on $h$.

  The second estimate holds, noting the asymptotics of $f'(x)/f(x)$ near $a$ in Lemma \ref{lem:3.5}, the first estimate of $g_h(x)$, and again $h(x) - \Exp_P[X] = \calO(1)$:
  \begin{align*}
    \frac{f'(x)}{f(x)} g_h(x)
    &= \left(\frac{2}{x-a} + \calO(1)\right)
    \left(\frac{1}{3} (x-a) (h(x) - \Exp_P[h]) + \calO((x-a)^2)\right) \\
    &= \frac{2}{3} (h(x) - \Exp_P[h]) + \calO((x-a)),
  \end{align*}
  where $\calO(\cdot)$ does not depend on $h$.
\end{proof}

With respect to a zero $a$ of $f$, we know by Theorem \ref{thm:gaps} that the first MIW sequence element $x_1$ to the right of $a$ satisfies $x_1 \rightarrow a$ as $N\uparrow\infty$.

\begin{lemma}\label{lem:6.16}
  Recall the setting of Proposition \ref{prop:wass}.  Suppose $a$ is finite.
  Then,
  \begin{align*}
    &\limsup_{N\uparrow\infty}\sup_{h\in \mathcal{L}} |\frac{f'(x_1)}{f(x_1)}g_h(x_1)| \leq \frac{2}{3}(\Exp_P[X]-a),\\
    &\limsup_{N\uparrow\infty}\sup_{h\in \mathcal{L}} \frac{|g_h(x_1)|}{x_{2} - x_{1}} \leq \frac{5}{3}(\Exp_P[X]-a).
  \end{align*}
\end{lemma}

\begin{proof}
  By Lemma~\ref{lem:6.10} and Lemma \ref{lem:bd-h}, we bound
  $\limsup_{x \to a} \sup_{h\in \mathcal{L}}\left|\frac{f'(x)}{f(x)}g_h(x)\right|
  \leq \sup_{h\in \mathcal{L}}\frac{2}{3}|h(a) - \Exp_P h| \leq \frac{2}{3}(\Exp_P[X] - a).$
  The first limit now follows as a consequence as $x_1\rightarrow a$ as $N\uparrow\infty$.

  Next, recall the MIW boundary condition at $x_0\leq a$:
  $\frac{1}{x_2-x_1} - \frac{1}{x_1-a} \leq \frac{1}{x_2-x_1} - \frac{1}{x_1-x_0} = \frac{f'(x_1)}{f(x_1)}$.
  Then,
  \begin{align*}
    \frac{|g_h(x_1)|}{x_2 - x_1}
    \leq \left(\frac{f'(x_1)}{f(x_1)} + \frac{1}{x_1 - a}\right) |g_h(x_1)|
    \leq \left|\frac{f'(x_1)}{f(x_1)} g_h(x_1)\right|
    + \frac{|g_h(x_1)|}{x_1 - a}.
  \end{align*}

  By the definition of $g_h$ and bound in Lemma \ref{lem:bd-h}, we have
  \begin{align*}
    \frac{|g_h(x_1)|}{x_1 - a}
    &= \frac{1}{(x_1 - a) f(x_1)} \left|\int_a^{x_1} f(t) (h(t) - \Exp_P h)\ dt\right| \\
    &\leq \frac{1}{(x_1 - a) f(x_1)}
    \int_a^{x_1} \sup_{s \in (a,x_1)} f(s) |h(s) - \Exp_P h|\ dt \\
    &= \sup_{s \in (a, x_1)} \frac{f(s)}{f(x_1)} |h(s) - \Exp_P h| \leq \sup_{s \in (a, x_1)} \frac{f(s)}{f(x_1)} (s + \Exp_P[X] - 2a).
  \end{align*}
  Observe as $f$ has a root of order exactly $2$ at $a$ that $\sup_{s\in (a, x_1)}f(s)/f(x_1) \rightarrow 1$ as $N\uparrow\infty$.
  Consequently, uniformly over $h\in \mathcal{L}$, $\limsup_{N\uparrow\infty} \frac{|g_h(x_1)|}{x_1 - a} \leq \Exp_P[X] - a$.

  Therefore, adding to the estimate from the first proven limit, we have
  \[
    \lim_{x_1 \to a} \sup_{h\in \mathcal{L}}\frac{|g_h(x_1)|}{x_2 - x_1}
    \leq \frac{5}{3}(\Exp_P[X]-a).
  \qedhere \]
\end{proof}

We now bound $|g_h''(x)|$ near the zero $a$ of $f$.

\begin{lemma}\label{lem:6.14}
  Recall the setting of Proposition \ref{prop:wass}.  Suppose $a$ is finite.

  Then, $\sup_{h\in \mathcal{L}}|g_h''(x)| = \calO(1)$ for $x$ near $a$.
\end{lemma}

\begin{proof}
  Taking each part of the expression for $g''_h$ in \eqref{eq:gh-split} in turn,
  noting Lemmas~\ref{lem:3.5}, \ref{lem:bd-h}, and~\ref{lem:6.10},
  \begin{align*}
    2\left(\frac{f'(x)}{f(x)}\right)^2 g_h(x)
    &= 2 \left(\frac {2}{x-a} + \calO(1)\right)
    \left(\frac{2}{3}(h(x) - \Exp_P[h]) + \calO((x-a))\right) \\
    &= \frac{8}{3} (x-a)^{-1} (h(x) - \Exp_P[h]) + \calO(1),
  \end{align*}

  \begin{align*}
    -\frac{f''(x)}{f(x)}g_h(x)
    &= -\frac{f''(a) + \calO((x-a))}{\frac{1}{2}f''(a)(x-a)^2 + \calO((x-a)^{3})}
    \left(\frac{1}{3} (x-a) (h(x) - \Exp_P[h]) + \calO((x-a)^2)\right) \\
    &= -(2(x-a)^{-2} + \calO((x-a)^{-1}))
    \left(\frac{1}{3} (x-a) (h(x) - \Exp_P[h]) + \calO((x-a)^2)\right) \\
    &= -\frac{2}{3} (x-a)^{-1} (h(x) - \Exp_P[h]) + \calO(1),
  \end{align*}
  and

  \begin{align*}
    -\frac{f'(x)}{f(x)} (h(x) - \Exp_P[h])
    &= -(2(x-a)^{-1} + \calO(1))(h(x) - \Exp_P[h]) \\
    &= -2(x-a)^{-1}(h(x) - \Exp_P[h]) + \calO(1).
  \end{align*}

  Then, the sum of these terms
  \begin{align*}
    g''_h(x)
    &= \frac{8}{3} (x-a)^{-1} (h(x) - \Exp_P[h]) + \calO(1) \\
    &\qquad -\frac{2}{3} (x-a)^{-1} (h(x) - \Exp_P[h]) + \calO(1) \\
    &\qquad -2(x-a)^{-1}(h(x) - \Exp_P[h]) + \calO(1) + h'(x) \ = \ \calO(1).
  \end{align*}
  The terms $\calO(\cdot)$ do not depend on $h$ and $h'(x) = \calO(1)$ for $h\in \mathcal{L}$.  Hence, 
  we can bound $|g_h''(x)| = \calO(1)$, uniformly in $h\in \mathcal{L}$, for $x$ near $a$.
\end{proof}

\begin{remark}
  \label{rmk:gh}\rm
  We note that analogous results in Lemmas \ref{lem:6.16} and \ref{lem:6.14} hold on $(a,b)$ where $-\infty\leq a<b<\infty$ with respect to finite right endpoints $b$ which are zeros of $f$, by reordering the sequence right to left or say by considering $f(b-x)$.  Recall the setting of Proposition \ref{prop:wass}.  We have
  \begin{align*}
    \limsup_{N\uparrow\infty}\sup_{h\in \mathcal{L}}|\frac{f'(x_N)}{f(x_N)} g_h(x_N)| &\leq \frac{2}{3}(b - \Exp_P[X]),\\
    \limsup_{N\uparrow\infty}\sup_{h\in \mathcal{L}}\frac{|g_h(x_N)|}{x_N-x_{N-1}} & \leq \frac{5}{3}(b-\Exp_P[X]),\\
    \sup_{h\in \mathcal{L}}|g''_h(x)| &= \calO(1) \ \ {\rm for \ } x\sim b.
  \end{align*}
\end{remark}

\subsection{Bound of $g''_h$ at infinity}

In this section, let $f$ be a higher energy function strictly positive on $(a, \infty)$,
where $a$ is a root of $f$ or $a=-\infty$.

To bound $g''_h$, it will be helpful to bound $|g_h|$ and $|x^2 g_h(x) + x(h(x)- \Exp_P[h])|$ in the next lemmas.

\begin{lemma}\label{lem:bdd-gh}
  We have
  \begin{align*}
    \limsup_{x \to \infty} \sup_{h\in \mathcal{L}}|g_h(x)| < \infty.
  \end{align*}
\end{lemma}

\begin{proof}
  Recall that $P=P_{a,\infty}$ is the distribution with density proportional to $f$ on $(a,\infty)$.
  Since $\int_a^\infty (h(x) - \Exp_P[h])f(x)dx =0$, and $f(t)/f(x)=p^2_\ell(t)e^{-\frac{1}{2}t^2}/p^2_\ell(x)e^{-\frac{1}{2}x^2}$,
  we may write
  \begin{align}
    \label{ghhelp}
    p_\ell(x)^2 e^{-\frac 12 x^2} g_h(x)
    = -\int_x^\infty (h(t) - \Exp_P[h]) p_\ell(t)^2 e^{-\frac 12 t^2}\ dt.
  \end{align}

  By Lemma \ref{lem:bd-h}, $|h(t) - \Exp_P[h]|\leq |t| + \Exp_P[|X|] \leq ct$ for large $t$.
  Since $p_\ell$ is an $\ell$ degree polynomial, we have also
  $c^{-1}t^\ell \leq p_\ell(t)\leq ct^\ell$ for large $t$.
  Moreover, by integration-by-parts, we have
  $\int_x^\infty t^{2\ell} te^{-\frac{1}{2}t^2}dt = x^{2\ell}e^{-\frac{1}{2}x^2} + \calO(x^{2\ell -1})$.
  Hence,
  $$g_h(x)
  \leq c^5 + \calO(x^{-1}),$$
  where $\calO(\cdot)$ does not depend on $h$.  The result follows. 
\end{proof}

\begin{lemma}\label{lem:bdd-snd}
  We have
  \begin{align*}
    \limsup_{x\rightarrow\infty}\sup_{h\in \mathcal{L}} |x^2 g_h(x) + x(h(x)-\Exp_P[h])| <\infty.
  \end{align*}
\end{lemma}

\begin{proof}
  We will need to evaluate $g_h$ more carefully than in Lemma \ref{lem:bdd-gh}.
  Suppose $\ell=0$ to begin.  
  Recall Mill's ratio:  $G(x)/e^{-\frac{1}{2} x^2} = \frac{1}{x} -\frac{1}{x^3} + o(\frac{1}{x^3})$ where
  $G(x) = \int_x^\infty e^{-\frac{1}{2}t^2}dt$.  Then, by \eqref{ghhelp}, and integration-by-parts,
  $$x^2 g_h(x) + x(h(x) - \Exp_P[h]) = \frac{x^2}{e^{-\frac{1}{2}x^2}} \Big[ -(h(x) - \Exp_P[h]) G(x) - \int_x^\infty h'(t)G(t)dt\Big] + x(h(x)-\Exp_P[h]).$$
  By Mill's ratio, and noting $|h(x)-\Exp_P[h]| = \calO(|x|)$ (Lemma \ref{lem:bd-h}) and $h'(x) = \calO(1)$, the desired limit holds.

  Suppose now $\ell\geq 1$.  Applying Lemma \ref{lem:poly-exp} with respect to the $2\ell$ degree polynomial $p_\ell$, we have
  \begin{align*}
    \int_{-\infty}^x p^2_\ell(t) e^{-\frac 12 t^2}\ dx
    = xq(x) e^{-\frac 12 x^2} - (q(0) - p(0)) \int_{-\infty}^x e^{-\frac 12 t^2}\ dt,
  \end{align*}
  where $q$ is degree $2\ell-2$ polynomial
  satisfying
  \begin{align}
    \label{q-exp}
    p^2_\ell(x) - p^2_\ell(0) = (1 - x^2) q(x) + xq'(x) - q(0).
  \end{align}

  Then, from \eqref{ghhelp}, as $p_\ell^2(t)e^{-\frac{1}{2}t^2} = \big[\int_{-\infty}^t p_\ell^2(s)e^{-\frac{1}{2}s^2}ds\big]'$,
  we have
  $$p_\ell^2(x) e^{-\frac{1}{2}x^2}g_h(x) = -\int_x^\infty (h(t)-\Exp_P[h]) \big[tq(t)e^{-\frac{1}{2}t^2}\big]' dt
  - (q(0) - p_\ell^2(0))\int_x^\infty (h(t) - \Exp_P[h]) e^{-\frac{1}{2}t^2}dt.$$
  Integrating the first term by parts,
  \begin{align*}
    -\int_x^\infty (h(t) - \Exp_P[h]) \left[
    tq(t) e^{-\frac 12 t^2} \right]'\ dt
  &= (h(x) - \Exp_P[h]) xq(x) e^{-\frac 12 x^2}
  + \int_x^\infty h'(t) tq(t) e^{-\frac 12 t^2}\ dt.
  \end{align*}
  This gives 
  \begin{align*}
   & p_\ell(x)^2 e^{-\frac 12 x^2} g_h(x)
   = (h(x) - \Exp_P[h]) xq(x) e^{-\frac 12 x^2}\\
   &\ \ \   + \int_x^\infty h'(t)
   tq(t) e^{-\frac 12 t^2}\ dt
   -(q(0) - p_\ell(0)^2)\int_x^\infty (h(t) - \Exp[h]) e^{-\frac 12 t^2}\ dt,
  \end{align*}
  and thus
  \begin{align*}
    g_h(x)
  &= x (h(x)-\Exp_P[h]) p_\ell(x)^{-2} q(x) \nonumber\\
  &\qquad - (q(0) - p_\ell(0)^2) p_\ell(x)^{-2}e^{\frac 12 x^2}
  \int_x^\infty (h(t) - \Exp_P[h]) e^{-\frac 12 t^2}\ dt \nonumber\\
  &\qquad + p_\ell(x)^{-2} e^{\frac 12 x^2}\int_x^\infty h'(t)
  tq(t) e^{-\frac 12 t^2}\ dt.
  \end{align*}
  Since, by Lemma \ref{lem:bd-h}, $|h(t) - \Exp_P[h]) = \calO(|t|)$, $h'(t) = \calO(1)$, and $q$ is an $2\ell -2$ degree polynomial, via integration-by-parts, we have
  \begin{align*}
    g_h(x)
    &= x(h(x)-\Exp_P[h])p_\ell(x)^{-2}q(x) + \calO(x^{-2}).
  \end{align*}
  Thus,
  \begin{align*}
    x^2 g_h(x) + x(h(x) - \Exp_P[h])
    &= \left(1 + \frac{x^2 q(x)}{p_\ell(x)^2}\right) x(h(x) - \Exp_P[h]) + \calO(1).
  \end{align*}
  Noting \eqref{q-exp},
  \begin{align*}
    1 + \frac{x^2 q(x)}{p_\ell(x)^2}
    = \frac{p_\ell(x)^2 + x^2 q(x)}{p_\ell(x)^2}
    = \frac{p_\ell(0)^2 + q(x) + xq'(x) - q(0)}{p_\ell(x)^2}
    = \calO(x^{-2}).
  \end{align*}
  We conclude, as by Lemma \ref{lem:bd-h}  $h(x)- \Exp_P[h] = \calO(x)$, that
  $|x^2 g_h(x) + x(h(x)-\Exp_P[h])| = \calO(1)$,
  where $\calO(\cdot)$ does not depend on $h$, giving the result.
\end{proof}

Finally, we come to the bound of $g''_h$.
\begin{lemma}
  \label{lem:gh2}
  We have
  \begin{align}
    \label{g''help}
    |g_h''(x)|
    \leq c_1 + c_2|g_h(x)| + |x^2 g_h(x) + x(h(x)- \Exp_P[h])|
  \end{align}
  for some constants $c_1$ and $c_2$ not depending on $h$.

  As a direct corollary of Lemmas \ref{lem:bdd-gh} and \ref{lem:bdd-snd}, we have
  $$
  \limsup_{x\rightarrow\infty} \sup_{h\in \mathcal{L}} |g''_h(x)| <\infty.
  $$
\end{lemma}

\begin{proof}
  Writing $\frac{f'(x)}{f(x)} = -x + \frac{2p'_\ell(x)}{p_\ell(x)}$, and noting that Hermite polynomials satisfy $p_\ell''(x) = xp_\ell'(x) - \ell p_\ell(x)$,
  we have that
  $$2\Big(\frac{f'(x)}{f(x)}\Big)^2 - \frac{f''(x)}{f(x)} = 6 \Big(\frac{p_\ell'(x)}{p_\ell(x)}\Big)^2
  - 6x\frac{p_\ell'(x)}{p_\ell(x)}
  + x^2 + 2\ell + 1.$$
  Hence, by direct computation with respect to \eqref{eq:gh-split},
  \begin{align*}
    &g_h''(x)
    &= \Big(6 \Big(\frac{p_\ell'(x)}{p_\ell(x)}\Big)^2
      - 6x\frac{p_\ell'(x)}{p_\ell(x)}
    + x^2 + 2\ell + 1 \Big) g_h(x)
    + \Big(x - 2\frac{p_\ell'(x)}{p_\ell(x)}\Big) (h(x)-\Exp_P[h])
    + h'(x).
  \end{align*}
  Then,
  \begin{align*}
    |g_h''(x)|
  &\leq |h'(x)|
  + \left|2\frac{p_\ell'(x)}{p_\ell(x)} (h(x)-\Exp_P[h])\right|\\
  &\ \ 
  + \Big|2\ell + 1 - 6x\frac{p_\ell'(x)}{p_\ell(x)}
  + 6\big(\frac{p_\ell'(x)}{p_\ell(x)}\big)^2\Big| |g_h(x)|
  + \Big|x^2 g_h(x) + x(h(x)-\Exp_P[h])\Big|.
  \end{align*}

  Note, by Lemma \ref{lem:bd-h} that $|h(x)-\Exp_P[h]| \leq |x| + \Exp_P[|X|]= \calO(x)$ for large $x$,
  and $|h'(x)| \leq 1$.
  Because $p_\ell$ is a polynomial,
  $\frac{p_\ell'(x)}{p_\ell(x)} = \calO\!\left(\frac 1x\right)$ for large $x$; it vanishes when $\ell=0$ as $p_0(x)\equiv 1$.
  Hence,  
  $\left|2\frac{p_\ell'(x)}{p_\ell(x)} (h(x)-\Exp_P[h])\right|
  = \calO(1)$ and
  $\left|2\ell + 1 - 6x\frac{p_\ell'(x)}{p_\ell(x)}\right|
  = \calO(1)$, and \eqref{g''help} follows.
\end{proof}

\subsection{Uniform bound of $g''_h$ in $(x_1, x_N)$}
Recall $f$ is a higher energy function, strictly positive on $(a, b)$ where $a$ is a zero of $f$ or $a=-\infty$, and $b$ is the next zero of $f$ or $b=\infty$.  Let $(x_n)_{n=1}^N$ be part of an MIW sequence contained in $(a,b)$.  

\begin{lemma}
  \label{lem:5.11}
  We have
  $$\sup_{x\in (x_1, x_N)}\sup_{h\in \mathcal{L}} |g''_h(x)| <\infty.$$
\end{lemma}

\begin{proof}
  Let $a<u<v<b$.  Recall the formula for $g''_h$ in \eqref{eq:gh-split}.  We have $h'(x) = \calO(1)$ and 
  $$g''_h(x) - h'(x) = 2\left(\frac{f'(x)}{f(x)}\right)^2 g_h(x)
  - \frac{f''(x)}{f(x)}g_h(x) - \frac{f'(x)}{f(x)} (h(x) - \Exp_P[h]).$$  As $|h(t) - \Exp_P[h]| \leq |t| + \Exp_P[|X|]$ by Lemma \ref{lem:bd-h}, we have
  $$\sup_{x\in (u,v)}\sup_{h\in \mathcal{L}}|g_h(x)| \leq \sup_{x\in (u,v)} f^{-1}(x)\int_a^x (|t| + \Exp_P[|X|])f(t)dt<\infty.$$ Also, $f'/f$ and $f''/f$ are bounded on $(u,v)$.  Hence, $\sup_{x\in (u,v)}\sup_{h\in\mathcal{L}}|g''_h(x)| <\infty$.

  Since $g''_h$ is bounded uniformly in $h$ near a root by Lemma \ref{lem:6.14} and Remark \ref{rmk:gh}, and at $\infty$ or $-\infty$ by symmetry of $f$ via Lemma \ref{lem:gh2}, the desired bound follows.
\end{proof}

\subsection{Wasserstein-$1$ distance bound on finite intervals}
Let $f$ be a higher energy function, strictly positive on $(a,b)$ where
where $-\infty < a < b < \infty$ and $a,b$ are zeros of $f$.
Let $P$ be the distribution for which $f$ is the density, and $Q$ be the empirical distribution of the part of the MIW sequence $(x_n)_{n=1}^N$ in $(a,b)$.

Recall that  $x_1 \to a$ and $x_N \to b$
as $N \to \infty$ by Theorem~\ref{thm:gaps}.

\begin{proposition}\label{prop:wass-finite}
  For all large $N$, the distance
  \begin{align*}
    d(Q, P) = \sup_{h\in \mathcal{L}}   |\Exp_Q[h] - \Exp_P[h]| =\calO\Big(\frac{1}{N}\Big).
  \end{align*}
\end{proposition}

\begin{proof}
  Note that $x_N-x_1\leq b-a$.
  Recall the intermediate continuous distribution $R$ in Proposition \ref{prop:wass1}, which bounds
  as a consequence $\sup_{h\in \mathcal{L}}|\Exp_R[h] - \Exp_Q[h]| =\calO(1/N)$.

  On the other hand, we bound $\sup_{h\in \mathcal{L}}|\Exp_R[h] - \Exp_P[h]|$ via the right-hand side of the estimate in Proposition \ref{prop:wass}.     Here, it does not matter what the value $\beta\in [0,1]$ is, but to be definite, we chose it as $\beta=0$.  By applying the bounds 
  $|\frac{g_h(x_1)}{x_2 - x_1}| \leq 2(b-a)$,
  $|\frac{g_h(x_N)}{x_N - x_{N-1}}| \leq 2(b-a)$,
  and $|\frac{f'(x_1)}{f(x_1)}g_h(x_1)| \leq b-a$
  in Lemma \ref{lem:6.16} and Remark \ref{rmk:gh}, and the bound of $\sup_{x\in (x_1, x_N)} |g''_h(x)|$ in Lemma \ref{lem:5.11}, we obtain $\sup_{h\in \mathcal{L}}|\Exp_R[h] - \Exp_P[h]| = \calO(1/N)$.

  Adding these two bounds, we obtain the desired estimate.
\end{proof}

\subsection{Wasserstein-$1$ distance bound on rays}
Let $f$ be a higher energy function, strictly positive on $(a,\infty)$ where
where $a$ is a zero of $f$.
Let $P$ be the distribution for which $f$ is the density, and $Q$ be the empirical distribution of the part of the MIW sequence $(x_n)_{n=1}^N$ in $(a,\infty)$.  

Recall that  $x_1 \to a$ and $x_N = \calO(\sqrt{\log N})$
as $N \to \infty$ by Theorem~\ref{thm:gaps}.

We comment, although the following is stated for rays $(a,\infty)$, by considering $f(a-x)$ we may also deduce the result for rays $(-\infty, a)$.

\begin{proposition}\label{prop:wass-half}
  For all large $N$, the distance
  \begin{align*}
    d(Q,P)= \sup_{h\in \mathcal{L}}    |\Exp_Q[h] - \Exp_P[h]|
    = \calO\Big(\frac{\sqrt{\log N}}{N}\Big).
  \end{align*}
\end{proposition}

\begin{proof}
  Note that $x_N-x_1= \calO(\sqrt{\log N})$.
  Recall again the intermediate continuous distribution $R$ in Proposition \ref{prop:wass1}, which bounds
  $\sup_{h\in \mathcal{L}}|\Exp_Q[h] - \Exp_R[h]| =\calO(\sqrt{\log N}/N)$.

  We now take $\beta = 1$ in the inequality in Proposition~\ref{prop:wass} of $|\Exp_R[h] - \Exp_P[h]|$.
  Since the right boundary condition holds,
  $\frac{1}{x_{N} - x_{N-1}} = -\frac{f'(x_N)}{f(x_N)}$,
  we have that
  $\frac{g_h(x_N)}{x_N-x_{N-1}} + \big(\frac{f'}{f}g_h\big)(x_N) = 0$.

  Hence, to bound the right-hand side in Proposition \ref{prop:wass}, we invoke Lemmas \ref{lem:6.16}, \ref{lem:5.11} to bound
  $\frac{g_h(x_1)}{x_2-x_1}$ and $\sup_{x\in (x_1, x_N)} |g''_h(x)|$.
  As a consequence, we obtain $\sup_{h\in \mathcal{L}}|\Exp_R[h]- \Exp_P[h]| = \calO(\sqrt{\log N}/N)$.

  Adding the bounds gives the result.
\end{proof}

\subsection{Wasserstein-$1$ distance bound for the Normal $\ell=0$ case)}

We rederive the bound found in \cite{mckeague_levin_2016, chen_thanh_2020, mckeague_swan_2021}, using the different  `density' approach.  
Let $f$ be the Normal density on $\R$,
$f(x) = \frac{1}{\sqrt{2\pi}} e^{-\frac 12 x^2}$.
Let $P$ be the associated distribution.
Let
$(x_n)_{n=1}^N$
be the associated MIW sequence satisfying the left boundary condition at $-\infty$,
and the right boundary condition at $\infty$, and
$Q$ its the empirical distribution.

Note that $x_N, |x_1| = \calO(\sqrt{\log N})$ by Theorem \ref{thm:gaps}.

\begin{proposition}[Normal]\label{lem:norm-bdd}
  We have
  \begin{align*}
    d(Q, P) = \sup_{h\in \mathcal{L}}|\Exp_Q[h] - \Exp_{P_N}[h]|
    = \calO\left(\frac{\sqrt{\log N}}{N}\right).
  \end{align*}
\end{proposition}

\begin{proof}
  Note that $x_N-x_1= \calO(\sqrt{\log N})$.
  Recall again the intermediate continuous distribution $R$ in Proposition \ref{prop:wass1}, which bounds
  $\sup_{h\in \mathcal{L}}|\Exp_Q[h] - \Exp_R[h]| =\calO(\sqrt{\log N}/N)$.

  Consider the bound in \eqref{eq:stein-bound} with the choice $\beta = 1$.
  Again, because the right boundary condition at infinity is
  $\frac{g_h(x_N)}{x_{N} - x_{N-1}} = -\left(\frac{f'}{f} g_h\right)(x_N)$,
  we need only bound
  (1) $\sup_{h\in \mathcal{L}}\frac{|g_h(x_1)|}{x_2 - x_1}=\calO(\sqrt{\log N})$, which requires a different argument than when the left endpoint of the interval is a zero of $f$, and (2)
  $\sup_{h\in \mathcal{L}}\sup_{x \in (x_1, x_N)} |g_h''(x)| =\calO(1)$ which has already been shown in Lemma \ref{lem:5.11}.

  To this end, from the MIW relation, as $f'(x)/f(x) = -x$, we have $x_2 - x_1 = -x_1^{-1}$.  Hence, we need to bound $|\frac{g_h(x_1)}{x_2-x_1}| = |x_1g_h(x_1)|$.
  Since $|h(t) - \Exp_P[h]| \leq |t| + \Exp_P[|X|]$, we have
  \begin{align*}
    |{x_1} g_h({x_1})|
    = \left| {x_1} e^{\frac 12 {x_1}^2} \int_{-\infty}^{x_1}
    e^{-\frac 12 t^2} (h(t) - \Exp_P[h])\ dt \right|
    \leq 
    2|x_1| e^{\frac 12 x_1^2} \int_{-\infty}^{x_1} e^{-\frac 12 t^2} |t|\ dt,
  \end{align*}
  for large $N$.
  As the term $x_1 \leq 0$, we have
  \begin{align*}
    |x_1g_h(x_1)|
                 & \leq 2{x_1} e^{\frac 12 {x_1}^2} \int_{-\infty}^{x_1} t e^{-\frac 12 t^2}\ dt 
                 = -2{x_1} = \calO(\sqrt{\log N}). \qedhere
  \end{align*}
\end{proof}

\subsection{Wasserstein-$1$ distance bound for $\ell\geq 0$:  Proof of Theorem \ref{thm:converge}}
\label{proofconvergesect}

We will break up $f$ into strictly positive regions.
We will then combine rates of convergence on each region,
via the following lemma.

\begin{proposition}\label{prop:wass-mix}
  Let $P_{n,k}$ and $P_k$ be probability distributions
  for $n \in \N$ and $0 \leq k \leq K$.
  Let $c_{n,k}>0$ be such that
  $\sum_{k=0}^K c_{n,k} = 1$.
  Assume $\max_{0\leq k \leq K} \int  |x|\ dP_k <\mu<\infty$.

  Suppose, uniformly in $k$, $P_{n,k} \to P_k$
  in the Wasserstein-$1$ metric with rate $r(n)$ and also
  $c_k$ is such that
  $|c_{n,k} - c_k| 
  \leq r(n)$.
  Define mixture probability distributions $M_n = \sum_{k=0}^K c_{n,k} P_{n,k}$,
  and 
  $M = \sum_{k=0}^K c_k P_k$.

  Then, $d(M_n, M) \leq (K + \mu)r(n)$.
\end{proposition}

\begin{proof}
  Write, noting $|h(x)-h(0)|\leq |x|$, 
  \begin{align*}
    d(M_n, M)
    &= \sup_{h\in \mathcal{L}} \left|\int h\ dM_n - \int h\ dM\right| \\
    &= \sup_{h\in \mathcal{L}} \left|\int (h - h(0))\ dM_n - \int (h-h(0))\ dM\right| \\
    &\leq \sum_{k=0}^K c_{n,k}
    \sup_h \left|\int (h-h(0)\ dP_{n,k} - \int (h-h(0)\ dP_k\right|
    + |c_{n,k} - c_k| \int |x|\ dP_k \\
    &\leq \sum_{k=0}^K c_{n,k} d(P_{n,k}, P_k) + r(n)\mu \ \leq \ 
    (K + \mu)r(n). \qedhere
  \end{align*}
\end{proof}

\begin{proof}[\textup{\textbf{Proof of Theorem~\ref{thm:converge}}}]
  When $\ell=0$, the result has been shown in Proposition \ref{lem:norm-bdd}.

  Suppose now $\ell\geq 1$ when the function $f$ has $\ell$ zeros, $r_1 < \cdots < r_\ell$.  Let $r_0=-\infty$ and $r_{\ell+1}=\infty$.
  Let $(x_n)_{n=1}^N$ be
  an $N$-element MIW sequence of $f$
  satisfying the left boundary condition at $-\infty$,
  and the right boundary condition at $\infty$.
  Suppose that in regions 
  $(r_k, r_{k+1})$
  there are $N_k=\lfloor N \int_{r_k}^{r_{k+1} f_\ell(x)\ dx} \rfloor$ many points for $0\leq k<\ell$ and $N_\ell=N-\sum_{k=0}^{\ell-1} N_k$ points in $(r_\ell, \infty)$.
  These sequences exist and are uniquely determined when $N_k\geq 1$ for $0\leq k\leq \ell$ when $\ell\geq 1$ and $N=N_0\geq 2$ when $\ell=0$
  by Theorem~\ref{thm:exist-unique}.

  On each region $(r_k, r_{k+1})$, consider the subsequence contained, reordered,
  $
  (x_1,\ x_2,\ \cdots,\ x_{N_k})$.  In the following, we will drop the subscript and call $N=N_k$.
  Denote as before by $Q$ the empirical distribution of this subsequence and $P$ the distribution with density proportional to $f$ on  
  $(r_k, r_{k+1})$.  Recall the `intermediate' continuous distribution $R$ in Proposition \ref{prop:wass1}.

  On finite intervals $(r_k, r_{k+1})$, 
  by Proposition \ref{prop:wass-finite} we have
  $d(Q,P)=    \sup_{h\in \mathcal{L}}|\Exp_{Q}[h] - \Exp_{P}[h]|
  \leq \frac{C}{N}$.

  On the rays $(r_\ell, \infty)$ and $(-\infty, r_1)$, 
  by Proposition~\ref{prop:wass-half} (and the comment before it), we have
  $d(Q, P)=  \sup_{h\in \mathcal{L}} |\Exp_{Q}[h] - \Exp_{P}[h]| \leq C\frac{\sqrt{\log N}}{N}$.

  Note that the proportion of points in each region, $c_{N,k}=N_k/N$,
  differs from the probability $c_k=\int_{r_k}^{r_{k+1}} f dx$ of that region by $\calO\left(\frac 1N\right)$.
  Thus, we recover the statement in the theorem, by applying Proposition~\ref{prop:wass-mix}.
\end{proof}

\section{Proof of Theorem \ref{thm:grad}: Stability of the MIW sequences}
\label{grad-sect}
We first give the proof of Theorem \ref{thm:grad}, and then make remarks on the behavior of the MIW sequence with respect to Maxwellian density, $\ell=1$, at the zero $t=0$, in Section \ref{remarks-stab-sect}.

 Let $\eta(x) =  \frac{f'(x)}{f(x)}$.
The following notation will also be useful:
Let
$\nabla_n k = k_n - k_{n-1}$ with respect to a sequence $k=(k_n)_{n=1}^N$, and 
$\zeta_n = 1/\nabla_n x= 1/(x_n - x_{n-1})$ with respect to a sequence $x=(x_n)_{n=1}^N$.
Then, $\eta(x_n)=\nabla_{n+1}\zeta = (x_{n+1}-x_n)^{-1} - (x_n - x_{n-1})^{-1}$ is a compact restatement of the MIW relation \eqref{eq:hdw-seq}.

The gradient of $H$ in \eqref{time-indep-H} is computed as follows. 

\begin{lemma}
\label{grad H calc}
We have that
 \begin{align}
    \label{partial-H}
    \partial_{x_n} H
    &= 2 x_n
    - 2 \frac{\eta(x_{n+1}) - \eta(x_n)}{(x_{n+1} - x_n)^2}
    + 2 \frac{\eta(x_n) - \eta(x_{n-1})}{(x_n - x_{n-1})^2}
    \end{align}
    on an MIW sequence $(x_n)_{n=1}^N$ for $f$.
    \end{lemma}
    
    \begin{proof}
 By the form of $H$, noting that only three terms in second sum of \eqref{time-indep-H} depend on $x_n$,
 \begin{align*}
 \frac{1}{2}\partial_{x_n}H &= x_n + (\nabla_n \zeta) \partial_{x_n} (\nabla_n \zeta) 
 + (\nabla_{n+1} \zeta) \partial_{x_n}(\nabla_{n+1}\zeta) + (\nabla_{n+2}\zeta) \partial_{x_n}(\nabla_{n+2}\zeta)\\
 &= x_n + (\nabla_n \zeta) \partial_{x_n} \zeta_n + (\nabla_{n+1}\zeta) \partial_{x_n}\zeta_{n+1} - (\nabla_{n+1}\zeta) \partial_{x_n} \zeta_n - (\nabla_{n+2}\zeta) \partial_{x_n} \zeta_{n+1}\\
 &=x_n - \zeta^2_{n}\nabla_n \zeta + \zeta^2_{n+1}\nabla_{n+1}\zeta + \zeta^2_n\nabla_{n+1}\zeta - \zeta^2_{n+1}\nabla_{n+2}\zeta\\
 &= x_n -\zeta^2_{n+1}\big(\nabla_{n+2}\zeta - \nabla_{n+1}\zeta\big) + \zeta^2_n\big(\nabla_{n+1}\zeta - \nabla_n\zeta\big).
 \end{align*}
 
 So far, we have not used the MIW relation $\nabla_{n+1}\zeta = \eta(x_n)$.  Now, inputting this, we obtain directly \eqref{partial-H}.
 \end{proof}

\begin{proof}[\textup{\textbf{Proof of Theorem~\ref{thm:grad}}}]
  Taylor expansion of $\eta$, at a point $x_n$ differing from a zero of $f$,
  gives
  \begin{align*}
    \frac{\eta(s) - \eta(x_{n})}{(s - x_{n})^2}
    &= \frac{\eta'(x_n)}{s - x_{n}}
    + \frac 12 \eta''(x_n) + r(s; x_n) 
  \end{align*}
  where the remainder
  $|r(s;x_n)| \leq |\eta'''(c)||s- x_n|$
  with $c$ between $s$ and $x_n$.

  Substituting into \eqref{partial-H}, with $s=x_{n+1}$ and also $s=x_{n-1}$, we obtain
  \begin{align}
    \partial_{x_n} H
    &= 2 x_n
    - 2\eta'(x_n)\left(\frac{1}{x_{n+1} - x_n} - \frac{1}{x_n - x_{n-1}}\right)
    - 2\eta''(x_n) - 2r(x_{n+1}; x_n) - 2r(x_{n-1}; x_n) \nonumber\\
    &= 2 x_n - 2\eta'(x_n)\eta(x_n) - 2\eta''(x_n) -2r(x_{n+1}; x_n) -2r(x_{n-1}; x_{n}).
        \label{partial-H1}
  \end{align}

  Recall that $\lim_{N \to \infty} x_{n(t)} = t$, and the `no gaps' property $\lim_{N\to \infty} x_{n+1}-x_n =0$ holds by Theorem \ref{thm:gaps}.  Therefore $\lim_{N\to \infty}r(x_{n(t)+1}; x_{n(t)})=0$, for $t \in \R$ not equal to a zero of $f$.
  Plugging in $n(t)$ for $n$ and taking the limit of both sides of \eqref{partial-H1} we have
  \begin{align*}
    \lim_{N \to \infty} \partial_{x_{n(t)}} H
    &= 2t - 2\eta'(t)\eta(t) - 2\eta''(t).
  \end{align*}

  Hence, for $t$ away from zeros of $f$,  $\lim_{N \to \infty} \partial_{n(t)} H = 0$
  exactly when $(x^2 + \eta^2(x) - 2\eta'(x))' = 0$ or when there is a constant $E$ such that
  \begin{equation}
    \label{H-ODE}
    x^2 = 2\eta'(x) + \eta^2(x) + E.
  \end{equation}
  However, the higher energy functions $f = cp^2_\ell(x)e^{-\frac{1}{2}x^2}$ satisfy \eqref{H-ODE}
  with $E = 4\ell + 2$, noting $p_\ell''(x) = xp_\ell' (x) -\ell p_\ell(x)$ and $\eta(x) = -x + \frac{2p_\ell'(x)}{p_\ell(x)}$, as
  \begin{align*}
    2\eta'(x) + \eta^2(x) + 4\ell + 2
    = 4 \frac{p_\ell''}{p_\ell} - 4x \frac{p_\ell'}{p_\ell} + x^2 + 4\ell
    = 4 \frac{x p_\ell' - \ell p_\ell}{p_\ell} - 4x \frac{p_\ell'}{p_\ell} + x^2 + 4\ell
    = x^2. 
  \end{align*}
Therefore, we conclude the proof of Theorem \ref{thm:grad}. \end{proof}

\subsection{Remarks on the form of stability}\label{remarks-stab-sect}

We demonstrate, with respect to the Maxwellian density $f$, $\ell=1$,
that the convergence in Theorem~\ref{thm:grad} does not hold with $t=0$,
the zero of $f$.
Note that $f'(x)/f(x) = -x + 2/x$.
Consider the MIW sequence $(x_n)_{n=1}^N$
with equal numbers of points to the left and right of zero.
By Corollary \ref{cor:Normalsymmetry}, such a sequence is symmetric. 
We will show in Lemma \ref{prop:graddiverge} that $\partial_{x_{n(0)+1}} H = -\partial_{x_{n(0)}} H$
diverges as $N \to \infty$.

First, although by Theorem \ref{thm:gaps},
$x_{n(0)+1}\rightarrow 0$ as $N\uparrow\infty$,
we give a rate in the following statement.

\begin{lemma}
  \label{lem:Max0}
  We have $x_{n(0)+1} = \calO(1/N^{1/3})$ as $N\uparrow\infty$.  
\end{lemma}

An intuitive understanding of the rate $x_{n(0)+1} = \calO(N^{-1/3})$ is as follows.
By the MIW ansatz, $x_{n(0)+1}$ should approximate a $1/(N+1)$ quantile,
that is one expects $\int_{x_{n(0)}}^{x_{n(0)+1}} f(t)\ dt = \calO(1/N)$
as indeed follows, since $x_{n(0)+1} = -x_{n(0)}$,
and $f(t) = \calO(t^2)$ for $t\downarrow 0$.

\begin{lemma}
  \label{prop:graddiverge}
  The gradient of $H$, with respect to the Maxwellian density $f$,
  at $x_{n(0)+1}$ is of lower and upper order $x_{n(0)+1}^{-3}$,
  and therefore is bounded below of order $N$, by 
  Lemma \ref{lem:Max0}, as $N\uparrow\infty$.
\end{lemma}

We first argue Lemma \ref{prop:graddiverge} and then Lemma \ref{lem:Max0}.

\begin{proof}[\textup{\textbf{Proof of Lemma~\ref{prop:graddiverge}}}]
  By explicit computation, applying the form of $f$ in \eqref{partial-H}, we have
  \begin{align*}
    \frac 12 \partial_{x_{n}} H
  &= x_{n}
  + \frac{x_{n+1} - \frac{2}{x_{n+1}} - x_{n} + \frac{2}{x_{n}}} {(x_{n+1} - x_{n})^2}
  - \frac{x_{n} - \frac{2}{x_{n}} - x_{n-1} + \frac{2}{x_{n-1}}} {(x_{n} - x_{n-1})^2} \\
  &= x_{n} + \frac{1}{x_{n+1} - x_{n}} - \frac{1}{x_{n} - x_{n-1}}
  - 2\frac{x_{n+1}^{-1} - x_{n}^{-1}}{(x_{n+1} - x_{n})^2}
  + 2\frac{x_{n}^{-1} - x_{n-1}^{-1}}{(x_{n} - x_{n-1})^2}.
  \end{align*}
  By the MIW relation, the right-hand side reduces to
  \begin{align}
  \label{grad1} 
    \frac{2}{x_n}
    - 2\frac{x_{n+1}^{-1} - x_{n}^{-1}}{(x_{n+1} - x_{n})^2}
    + 2\frac{x_{n}^{-1} - x_{n-1}^{-1}}{(x_{n} - x_{n-1})^2}.
  \end{align}

    After algebra,
    \begin{align*}
      -\frac{x_{n+1}^{-1} - x_{n}^{-1}}{(x_{n+1} - x_{n})^2}
    &= \frac{x_{n+1} - x_n}{x_n x_{n+1}(x_{n+1} - x_{n})^2} \ = \ \frac{1}{x_n^2} \frac{x_n}{x_{n+1}(x_{n+1} - x_n)} \\
    &= \frac{1}{x_n^2} \frac{x_{n+1} - (x_{n+1} - x_n)}{x_{n+1}(x_{n+1} - x_n)} \ = \
     \frac{1}{x_n^2} \left( \frac{1}{x_{n+1} - x_n} - \frac{1}{x_{n+1}} \right),
    \end{align*}
    and likewise
    \begin{align*}
      \frac{x_{n}^{-1} - x_{n-1}^{-1}}{(x_{n} - x_{n-1})^2}
    &= \frac{1}{x_n^2} \left( -\frac{1}{x_n - x_{n-1}} - \frac{1}{x_{n-1}} \right).
    \end{align*}
    By these relations, and the MIW relation again, \eqref{grad1} becomes
    \begin{align*} 
  & \frac{2}{x_n}
    + \frac{2}{x_n^2} \left( \frac{1}{x_{n+1} - x_n} - \frac{1}{x_{n+1}}
    -\frac{1}{x_n - x_{n-1}} - \frac{1}{x_{n-1}} \right) 
    \ = \ -\frac{2}{x_n^2} \left( \frac{1}{x_{n+1}} - \frac{2}{x_n} + \frac{1}{x_{n-1}} \right) \\
    &\ \ \ \ = \frac{2}{x_n^3} \left( 1 - \frac{x_n}{x_{n+1}} - \Big(\frac{x_n}{x_{n-1}} - 1\Big) \right) 
    \ = \ \frac{2}{x_n^3} \left( \frac{x_{n+1} - x_n}{x_{n+1}} - \frac{x_n - x_{n-1}}{x_{n-1}} \right).
    \end{align*}
  
    Choose $n = n(0) + 1$.
    By symmetry, $x_{n(0)} = -x_{n(0)+1}$.  So,
    \begin{align*}
      \frac{2}{x_{n(0)+1}^3} \left( \frac{x_{n(0)+2} - x_{n(0)+1}}{x_{n(0)+2}} - \frac{x_{n(0)+1} - x_{n(0)}}{x_{n(0)}} \right)
    &= \frac{2}{x_{n(0)+1}^3} \left( \frac{x_{n(0)+2} - x_{n(0)+1}}{x_{n(0)+2}} + \frac{x_{n(0)+1} + x_{n(0)+1}}{x_{n(0)+1}} \right) \\
    &= \frac{2}{x_{n(0)+1}^3} \left( 3 - \frac{x_{n(0)+1}}{x_{n(0)+2}} \right).
    \end{align*}
   Then,  
      $\partial_{n(0)+1} H
      = \frac{4}{x_{n(0)+1}^3} \left( 3 - \frac{x_{n(0)+1}}{x_{n(0)+2}} \right)$.
    Since $0 < x_{n(0)+1} < x_{n(0)+2}$, we conclude as desired
    \[
      8 x_{n(0)+1}^3
      < \partial_{x_{n(0)+1}} H
      < 12 x_{n(0)+1}^3. \qedhere \]
\end{proof}

\begin{proof}[\textup{\textbf{Proof of Lemma~\ref{lem:Max0}}}]
  Multiply the MIW relation by $x_n^2$ and sum over $n$ to obtain
  \begin{align*}
    \sum_{n=k+1}^N \frac{1}{x_n^2}\Big(\frac{1}{x_{n+1}-x_n} - \frac{1}{x_n-x_{n-1}}\Big)
    = \sum_{n=k+1}^N \frac{1}{x_n^2} \Big(\frac{2}{x_n} - x_n\Big).
  \end{align*}
  Summing-by-parts on the left-hand side, and then rearranging, we obtain
  \begin{align}
    \label{maxhelp}
    \frac{1}{x_{k+1}^2}\frac{1}{x_{k+1} - x_k} = \sum_{n=k+1}^N \Big(\frac{1}{x_n^2} - \frac{1}{x_{n+1}^2}\Big) \frac{1}{x_{n+1} - x_n} - \sum_{n=k+1}^N \frac{1}{x_n^2} \Big(\frac{2}{x_n} - x_n\Big).
  \end{align}
  Note 
  $$\Big(\frac{1}{x_n^2} - \frac{1}{x_{n+1}^2}\Big) \frac{1}{x_{n+1} - x_n} = \frac{x_n-x_{n+1}}{x_n^2x_{n+1}^2} + \frac{2}{x_n^2x_{n+1}},$$
  and $\frac{-1}{x_n^2} \big(\frac{2}{x_n} - x_n\big) = \frac{1}{x_n}-\frac{2}{x_n^3}$.  We also note by symmetry that $x_{n(0)+1} -x_{n(0)} = 2x_{n(0)+1}$.

  Putting these things together, with $k= n(0)$, \eqref{maxhelp} multiplied by $2x_{k+1}^3$ becomes
  \begin{align*}
    1 = 2x_{k+1}^3 \sum_{n=k+1}^N \frac{1}{x_n} + 2x_{k+1}^3\sum_{n=k+1}^N \frac{1}{x_n^2} \Big( \frac{2}{x_{n+1}} - \frac{2}{x_n} + \frac{x_n - x_{n+1}}{x_{n+1}^2}\Big).
  \end{align*}

  Since the HDW sequence $(x_n)_{n=1}^N$ is increasing, and $x_n>0$ for $n\geq k+1=n(0)+1$, we have
  $$x^3_{k+1}\sum_{n=k+1}^N \frac{1}{x_{n+1}^2} \Big|\frac{2}{x_{n+1}} - \frac{2}{x_n}\Big|
  \leq x_{k+1}\sum_{n=k+1}^N \Big(\frac{2}{x_n} - \frac{2}{x_{n+1}}\Big) = \frac{2x_{k+1}}{x_{k+1}} = 2$$
  and
  $$x^3_{k+1}\sum_{n=k+1}^N \frac{1}{x_{n+1}^2} \frac{|x_n-x_{n+1}|}{x_{n+1}^2} \leq x_{k+1}\sum_{n=k+1}^N \frac{x_{n+1}-x_n}{x_nx_{n+1}} = x_{k+1}\sum_{n=k+1}^N \Big(\frac{1}{x_n}-\frac{1}{x_{n+1}}\Big) = 1.$$

  Therefore, we have
  $$1 + 6 \geq 2Nx_{k+1}^3 \cdot \frac{1}{N}\sum_{n=k+1}^N \frac{1}{x_n}.$$
 Moreover, we conclude
  $7\geq 2Nx_{k+1}^3\cdot \frac{1}{2}\int_0^\infty \frac{1}{t} \sqrt{\frac{8}{\pi}}t^2 e^{-\frac{1}{2}t^2}dt$ by Fatou's lemma, with respect to the weak convergence in Theorem \ref{thm:converge},
  for all large $N$.  The result now follows. 
\end{proof}

\appendix

\section{Properties of higher energy functions}

We consider several properties of higher energy functions that will be needed.

\begin{lemma}\label{lem:dbl-zro}
  Suppose $f$ is an $n$th order higher energy function.
  Then $f$ has exactly $n$ real roots.
  These roots are distinct, real, and have order two.

  Moreover, near a root $a$, $f(x) = \frac 12 f''(a) (x-a)^2 + \calO((x-a)^3)$
  and $f'(x) = f''(a) (x-a) + \calO((x-a)^2)$,
  where $f''(a) \neq 0$.
\end{lemma}

\begin{proof}
  The Hermite polynomial $p_n(x)= e^{\frac 12 x^2} g_n(x)$, where $g_n(x) = \frac{d^n}{dx^n} e^{-\frac 12 x^2}$, has exactly $n$ simple real roots, as $g_n$ has exactly $n$ simple real roots, which can be shown by induction.
  These are exactly the roots of $f(x) = cp^2_n(x) g_0(x)$ which are of order $2$.  The remaining statements follows by standard Taylor expansion.
\end{proof}

\begin{lemma}\label{lem:higher-order}
  The higher energy functions are log-concave,
  in the sense that $(\log f)''(x) < 0$ when $(\log f)''(x)$ is defined.
\end{lemma}

\begin{proof}
  Let $f(x)$ be a higher energy function.
  The log-derivative $(\log f)'(x) = 2\frac{p'_n(x)}{p_n(x)} - x$.
  The term $-x$ is decreasing on $\R$.
  We will show that $\frac{p'_n(x)}{p_n(x)}$ is decreasing as well,
  for $x \in \R$ such that $p_n(x) \neq 0$.
  In the case $n = 0$, the normal case,
  $p_n$ is constant, so it is decreasing.

  Suppose the order of $p_n$ is $n\geq 1$.
  We factor $p_n(x)$,
  $ p_n(x) = \prod_{k=1}^n (x - r_i)$,
  where $r_i$ are the roots of $p_n$.
  Then,
  $p_n'(x) = \sum_{j=1}^n \prod_{\substack{k=1 \\ j \neq k}}^n (x - r_i)$.
  Therefore, the quotient
  \begin{align*}
    \frac{p_n'(x)}{p_n(x)}
    &= \frac{\sum_{j=1}^n \prod_{\substack{k=1 \\ j \neq k}}^n (x - r_i)}
    {\prod_{k=1}^n (x - r_i)}
    = \sum_{j=1}^n \frac{\prod_{\substack{k=1 \\ j \neq k}}^n (x - r_i)}
    {\prod_{k=1}^n (x - r_i)}
    = \sum_{j=1}^n \frac{1}{x-r_j}.
  \end{align*}
  This is a sum of decreasing functions, and
  so $\frac{p_n'(x)}{p_n(x)}$ is decreasing.
\end{proof}

\begin{lemma}
  \label{lem:3.5}
  Let $f$ be a higher energy function defined on $(-\infty, \infty)$.
  The limits $\lim_{x \to -\infty} \frac{f'(x)}{f(x)} = \infty$,
  and $\lim_{x \to \infty} \frac{f'(x)}{f(x)} = -\infty$.
  As well, if $f(r) = 0$ for some $r \in \R$,
  then $\frac{f'(x)}{f(x)} = \frac{2}{x-r} + O(1)$ as $x\rightarrow r$.  In particular, $\lim_{x \to r^-} \frac{f'(x)}{f(x)} = -\infty$,
  and $\lim_{x \to r^+} \frac{f'(x)}{f(x)} = \infty$.
\end{lemma}

\begin{proof}
  As shown in the proof of Lemma~\ref{lem:higher-order},
  we can write
  \begin{align*}
    \frac{f'(x)}{f(x)}
    = (\log f)'(x)
    = 2\frac{p'_n(x)}{p_n(x)} - x
    = 2\sum_{j=1}^n \frac{1}{x-r_j} - x.
  \end{align*}
  As $x \to \pm \infty$, $f'(x)/f(x) \rightarrow \mp \infty$.

  For a root $r$ of $f$,
  the term $\frac{1}{x - r}$ is the only unbounded term near $r$.
  Hence, as $x\rightarrow r^-$, $f'(x)/f(x) \rightarrow -\infty$.  Likewise, as $x\rightarrow r^+$, $f'(x)/f(x) \rightarrow \infty$.
\end{proof}

Although we will apply the following lemma for $p=p_\ell$, the integration-by-parts formula is stated more generally.
\begin{lemma}\label{lem:poly-exp}
  Given a polynomial $p$ of degree $n \geq 2$, there is a polynomial $q$ of degree $n-2$ 
	satisfying
  \begin{align}
	\label{qhelp}
    p(x) - p(0) = (1 - x^2) q(x) + xq'(x) - q(0),
  \end{align}
	so that
  \begin{align*}
    \int_{-\infty}^x p(t) e^{-\frac 12 t^2}\ dx
    = xq(x) e^{-\frac 12 x^2} - (q(0) - p(0)) \int_{-\infty}^x e^{-\frac 12 t^2}\ dt.
  \end{align*}
\end{lemma}

\begin{proof}
  If $q$ satisfies \eqref{qhelp},
	then
  \begin{align*}
    p(x) e^{-\frac 12 x^2}
    &= q(x) e^{-\frac 12 x^2}
    + x q'(x) e^{-\frac 12 x^2}
    - x^2 q(x) e^{-\frac 12 x^2}
    - (q(0) - p(0)) e^{-\frac 12 x^2}
  \end{align*}
from which the integral relation in the lemma holds.

  However, a polynomial solution of degree $n-2$
  of \eqref{qhelp},
	$q(x) = \sum_{k=0}^{n-2} b_k x^k$,
  may be found by finding iteratively the coefficients $(b_k)_{k=0}^{n-2}$,
  given that $p$ is in form $p(x) = \sum_{k=0}^n a_k x^k$.
\end{proof}

\section{Non-existence and non-uniqueness of MIW sequences}
\label{non-exist-sect}
We discuss that it is not given that there is a unique MIW sequence matched to boundary conditions for say any smooth, strictly positive $f$.

\begin{lemma}
  There is a smooth, strictly positive function $f$ on $\R$
  such that there is no MIW sequence of $f$ with $N \geq 2$ points
  which satisfies the left boundary condition at $-\infty$.
\end{lemma}

\begin{proof}
  The left boundary condition at $-\infty$ is
  $\frac{1}{x_2 - x_1} = \frac{f'(x_1)}{f(x_1)}$,
  where $x_2 > x_1$.
  This has no solution if $\frac{f'(x_1)}{f(x_1)} = 0$, for instance if $f(x) \equiv 1$.
\end{proof}

\begin{lemma}
  There is a smooth, strictly positive integrable function $f$ on $\R$
  there are two distinct two-element MIW sequences of $f$
  which satisfy the left boundary condition at $-\infty$,
  and the right boundary condition at $\infty$.
\end{lemma}

\begin{proof}
  We will create an $f$ such that
  both $(-2, -1)$ and $(1, 2)$ are MIW sequences of $f$
  satisfying the left boundary condition at $-\infty$,
  and the right boundary condition at $\infty$.
  Taking care of the positive side first,
  we require that $\frac{1}{2 - 1} = \frac{f'(1)}{f(1)}$
  and that $-\frac{1}{2 - 1} = \frac{f'(2)}{f(2)}$.
  We can thus choose $f(x) = x$ in a neighborhood of one,
  and $f(x) = 3 - x$ is a neighborhood of two.

  Likewise to satisfy the requirements at $-2$ and $-1$,
  we can choose $f(x) = -x$ in a neighborhood of $-1$,
  and $f(x) = 3 + x$ is a neighborhood of $-2$.
  Because $f$ is specified only in the neighborhoods of finitely many points,
  and is smooth and strictly positive at those points,
  we can extend $f$ to a smooth, strictly positive integrable function on $\R$.
\end{proof}

\medskip
\noindent
{\bf Funding.}  This work was partially supported by ARO-W911NF-18-1-0311.

\end{document}